\documentclass[11pt]{article} 

\usepackage{latexsym,amsmath,url,caption,epsfig}
\usepackage{textcomp}
\usepackage{amsfonts,euscript}
\usepackage{graphicx}
\usepackage{amsmath}
\usepackage{amssymb}
\usepackage{amsfonts}
\usepackage{amsthm}
\usepackage{mathrsfs}
\usepackage[all]{xy}
\usepackage{epstopdf}
\usepackage{graphics}
\usepackage{subfigure}
\usepackage{rotating}
\usepackage{appendix}
\usepackage{multirow,rotating}
\usepackage{url}
\usepackage{setspace}
\usepackage{listings}
\usepackage{color}
\usepackage{tikz}
\usepackage[breaklinks=true]{hyperref}
\usepackage[english]{babel}
\usepackage{datetime}
\usepackage{booktabs}
\usepackage{rotfloat}
\usepackage[numbers]{natbib}
\usepackage{algorithm}
\usepackage{algorithmicx}
\usepackage{algpseudocode}
\usepackage{framed}
\usepackage{authblk}

\graphicspath{{pic/}}

\DeclareMathOperator{\E}{\mathbb{E}}

\usepackage{endnotes}

\usepackage{natbib}
 \bibpunct[, ]{(}{)}{,}{a}{}{,}%

\usepackage{xcolor}
\hypersetup{
    colorlinks,
    linkcolor={red!50!black},
    citecolor={blue!50!black},
    urlcolor={blue!80!black}
}

\makeatletter
\newenvironment{breakablealgorithm}
  {
  \begin{center}
     \refstepcounter{algorithm}
     \hrule height.8pt depth0pt \kern2pt
     \renewcommand{\caption}[2][\relax]{
      {\raggedright\textbf{\ALG@name~\thealgorithm} ##2\par}%
      \ifx\relax##1\relax 
         \addcontentsline{loa}{algorithm}{\protect\numberline{\thealgorithm}##2}%
      \else 
         \addcontentsline{loa}{algorithm}{\protect\numberline{\thealgorithm}##1}%
      \fi
      \kern2pt\hrule\kern2pt
     }
  }{
     \kern2pt\hrule\relax
  \end{center}
  }
\usepackage{makecell}

\newenvironment{assumption'}[1]
  {%
   \addtocounter{assumption}{-1}%
   \begin{assumption}}
  {\end{assumption}}

\newtheorem{theorem}{Theorem}

\newtheorem{lemma}{Lemma}
\newtheorem{proposition}{Proposition}
\newtheorem{assumption}{Assumption}
\newtheorem{definition}{Definition}
\newtheorem{example}{Example}

\allowdisplaybreaks

\begin{document}


\title{Selecting the Best Optimizing System\footnote{We gratefully acknowledge helpful discussions with Steve Chick, Weiwei Fan, Shane Henderson, Jeff Hong, Peter Glynn, Guanghui Lan, and Barry Nelson. All errors are ours.}}


\author[1]{Nian Si\thanks{niansi@ust.hk}}
\author[2]{Yifu Tang\thanks{yifutang@berkeley.edu}}
\author[2]{Zeyu Zheng\thanks{zyzheng@berkeley.edu}}

\affil[1]{Department of Industrial Engineering and Decision Analytics,
Hong Kong University of Science and Technology}

\affil[2]{Department of Industrial Engineering and Operations Research, University of California, Berkeley, CA, USA}  


\maketitle

\begin{abstract}
We formulate  \textit{selecting the best optimizing system} (SBOS) problems and provide solutions for those problems. In an SBOS problem, a finite number of systems are contenders. Inside each system, a continuous decision variable affects the system's expected performance where the exact optimal choice is not accessible within finite number of samples. An SBOS problem compares different systems based on their expected performances under their own optimally chosen decision to select the best, without advance knowledge of the expected performances of the systems or the optimizing decision inside each system. We design easy-to-implement algorithms that adaptively choose a system and a choice of decision to evaluate the noisy system performance, sequentially eliminate inferior systems, and eventually recommend a system as the best after spending a user-specified budget. The proposed algorithms integrate the stochastic gradient descent method and the sequential elimination method to simultaneously exploit the structure inside each system and make comparisons across systems. For the proposed algorithms, we prove exponential rates of convergence to zero for the probability of false selection as the budget grows to infinity. We conduct three numerical examples that represent three practical cases of SBOS problems. Our proposed algorithms demonstrate consistent and stronger performances in terms of the probability of false selection over benchmark algorithms under a range of problem settings and sampling budgets.
\end{abstract}







\section{Introduction}

The need to select a system with the best mean system performance among a number of different systems naturally arises in various decision-making problems. The decision maker is typically able to generate or collect unbiased noisy random samples of the expected performance for each system in contention. The task of selecting the best system in a statistically principled way is a fundamental research problem in several growing research areas. In the area of stochastic simulation, this research problem is referred to as \textit{Ranking and Selection} or \textit{Selecting the Best System}; see \cite*{kim2006selecting}, \cite*{chick2006subjective}, \cite*{hunter2017parallel} and \cite*{hong2021review} for comprehensive reviews and see \cite*{chick2005selection,lan2010confidence,waeber2010performance}, \cite*{luo2015fully}, \cite*{fan2020distributionally,shen2021ranking} for various applications in healthcare decisions, production management, financial risk evaluation and decisions. For most \textit{Selecting the Best System} problems, noisy random samples are generated from running costly stochastic simulations, where the simulation model is often built to represent real complicated systems or new systems that have yet to be developed. The task of selecting the best system also appears in experimental designs for clinical trials or A/B testing (see \cite*{johari2021always} and \cite*{chick2021bayesian} for example), where the noisy random samples are collected from running costly treatment experiments on individuals. Both the \textit{Selecting the Best System} literature and the clinical trials literature point back to \cite*{bechhofer1954single} and \cite*{bechhofer1995design}, and both literatures aim at selecting the best or better system in a statistically principled way. The two literatures share the same notion that samples are noisy and costly to generate or collect, despite of the difference on how the samples are generated or collected.

Most existing literature on selecting the best system (i.e., selecting a system with the best mean system performance) assumes that one has access to independent unbiased noisy samples of the system performance for each system in comparison. However, this access can be unavailable for problems where each system in comparison involves optimizing a decision inside the system. For instance, such a phenomenon arises in the following application examples.

\smallskip

\begin{itemize}
    \item\textit{Medication and Healthcare.} A pharmaceutical factory develops two new drugs for treating insomnia in specific patients. To choose which drug to produce and promote, it must compare their effectiveness (see, for example \cite*{erman2006efficacy,verweij2020randomized}). Each drug is treated as a “system,” where performance depends on the expected effect at its optimal dose. Since dosage is a continuous variable and the optimal level is unknown in advance, the decision-maker lacks an unbiased estimator of this maximum effect. Thus, selecting the best drug requires not only comparing drugs but also solving an inner optimization problem to determine each drug’s optimal dosage. 
    \item\textit{Simulation Optimization for Queuing Service Systems.}
    In non-stationary queuing service systems, the manager uses simulation optimization to design staffing plans under resource constraints. With $N$ staff members, the goal is to choose the plan that maximizes expected performance (e.g., revenue, reduced abandonment). One example of such decision is to set a price for the service (\cite*{kim2018value}, \cite*{lee2019pricing}). Each staffing plan is treated as a “system”, in which additional decisions—such as setting service prices—must also be optimized. Since the optimal price varies by plan and is not known beforehand, comparing staffing plans requires first optimizing pricing within each system, then selecting the plan with the best resulting performance.

    \item\textit{Data-driven Revenue Management and Product Selection.} A platform wants to select the best product from $N$ candidates to post and sell, but it does not know the demand distribution for each product. Each day, it can display only one product, observe its random demand, and treat that product as a “system”. In each system, the platform must also make an inventory decision to maximize expected profit, though the optimal decision is unknown and improves with more data. Given a budget of $T$ days, the challenge is to allocate sampling across products to identify the one with the best optimized expected profit.
\end{itemize}

\smallskip



Formally, these aforementioned applications motivate us to define and analyze a class of problems that we call \textit{selecting the best optimizing system} (SBOS). The SBOS problems have two layers of optimization. The outer-layer optimization involves a selection from a finite number of systems. For each system, there is an inner-layer optimization inside the system, where there is a continuous decision variable that affects the expected system performance. The inner-layer optimization decides the optimizing performance of each system by choosing the best decision variable inside the system. The outer-layer optimization selects the best system according to the  optimizing performance.

In this work, we consider a fixed-budget formulation of the SBOS problems. That is, there is a given budget of samples and one has the ability to sequentially decide how the samples are allocated to different systems. Once the sampling budget is exhausted, based on all observations, a recommendation needs to be made on which system has the best optimizing performance. The goal is to design easy-to-implement algorithms that are allowed to sequentially allocate the samples and end up with a recommendation on the best optimizing system. The metric to evaluate how good an algorithm is by the probability of false selection (PFS) given a fixed budget. The designed algorithms are desired to demonstrate good empirical performances and, enjoy a theoretical guarantee on the upper bound of PFS under a given budget.

A major challenge in designing algorithms for SBOS problems is that system comparisons depend on optimizing performance, which is unknown a priori. To know it exactly, the inner optimization must be fully solved—but with finite samples, unbiased estimates are generally impossible. Unlike classical best-system selection, where unbiased independent samples are available, SBOS only allows biased estimates. However, with intelligent sampling, the bias can be reduced. Thus, algorithm design must consider both how optimized a system already is and the variability of its performance.
Another challenge—and opportunity—comes from the structure of the inner optimization. Properties like convexity can guide the outer optimization, influencing sampling decisions. Effectively leveraging these structures poses difficulties but also enables stronger performance if incorporated into algorithm design.

We summarize our main contributions in the following subsection.

\subsection{Main Contributions}
First, motivated by applications in simulation optimization, data-driven stochastic optimization and medication decisions, we formulate a new class of problems named \textit{Selecting the Best Optimizing System} (SBOS). In SBOS problems, systems are compared based on their optimizing performances, which involve an inner-layer of optimization in addition to the standard selection optimization. We consider a fixed-budget formulation where the decision maker sequentially decides how to spend the sampling budget, based on sequentially observed sampling outcomes. 
The SBOS problems naturally incorporate two streams of settings where the sampling cost is because of expensive simulation and is because of expensive real experiments and data collection.


Second, we propose simple algorithms with exponential-rate performance guarantees: as the sample budget grows to infinity, the probability of false selection decays exponentially with an explicit positive rate. Unlike classical best-selection problems, where performance is estimated from i.i.d. samples, our setting requires inner-layer optimization to evaluate each system. To address this, our design combines stochastic gradient descent and sequential elimination, exploiting the structure of the inner optimization while making outer-layer comparisons. The algorithms carefully control bias and variance in performance estimation, and our analysis accounts for both while leveraging structural information. We prove that the exponential-rate guarantees hold for convex simulation optimization and broader data-driven stochastic optimization, covering key SBOS problem classes. We expect these results to extend to a wider range of simulation problems with both continuous and categorical variables.


 Third, we conduct comprehensive numerical studies for the SBOS problems, including three practical applications. The first application is an optimal staffing and pricing problem in a non-stationary queueing system. The second application is an optimal dosage finding problem in the selection of the best drug. The third application is a data-driven newsvendor problem in the selection of the best product. We compare our proposed algorithms to the uniform sampling method and the Optimal Computing Budget Allocation (OCBA) method. We demonstrate that our method achieve the lowest probability of false selection in all problem settings.


\subsection{Connections to Related Literature}

Our work is closely connected to the literature on fixed-budget ranking and selection (R\&S) problems. Instead of using the term R\&S, we adopt in this work the terminology of {selecting the best system}, which has been an equivalent or even slightly more precise notion when one does not rank the systems but only selects the best (see \cite*{kim2006selecting,hong2007selecting}). The optimal computing budget allocation (OCBA) procedure proposed by  \cite*{chen2000simulation} and its sequential version is among the most famous algorithm for fixed-budget R\&S problems. \cite*{glynn2004large} establishes a rigorous guarantee for the OCBA procedure using a large deviation principle. \cite*{wu2018analyzing} takes a closer look at sequential OCBA algorithms and demonstrate exponential decaying rate for the Probability of False Selection (PFS) as the budget goes to infinity. We refer to the references within \cite*{hunter2017parallel}, \cite*{wu2018analyzing} and \cite*{hong2021review} for comprehensive reviews of fixed-budget and fixed-confidence R\&S work. Besides frequentist approaches,  \cite*{frazier2009knowledge}, \cite*{chick2010sequential,chick2012sequential,ryzhov2016convergence}, \cite*{chen2019complete}, \cite*{russo2020simple}, \cite*{li2020context}, \cite*{eckman2022posterior}, \cite*{song2025optimizing} and references within for the use or discussion of Bayesian methods. We also refer to \cite{zhang2023asymptotically} for selecting top-$m$ alternatives.

We define a class of \textit{fixed-budget SBOS problems}.  SBOS is different from standard R\& S because unbiased samples of a system’s performance are unavailable---each system involves an inner optimization that cannot be exactly solved with finite samples. Still, the OCBA method used in R\&S can be adapted. For example, in simulation optimization with $K$ systems, each requiring an optimal price $p_i^*\in[0,1]$, one can discretize each system into $M$ subsystems with different prices, turning the two-layer SBOS problem into a standard R\&S problem with $KM$ systems where OCBA applies. This ``break-down-and-then-OCBA'' approach faces two issues: (1) for continuous inner variables, a fixed $M$ may exclude the true optimizer, and dynamically scaling $M$ complicates design; (2) it ignores structural information of the inner optimization (e.g., convexity), reducing to grid search. We instead propose a sequential elimination method that avoids system break-down and overcomes these challenges, supported by extensive experiments comparing it to OCBA.

Our work is also related to \cite*{fan2020distributionally}. They discuss the robust selection of the best system (RSB), where the probability distributions associated with each system are not exactly known but may come from a set consisting of a finite number of options. 
In \cite*{fan2020distributionally}, the best system is the one possessing the best worst-case performance. 
The algorithm design is different: their work is based on the indifference zone-free sequential procedure in \cite*{fan2016indifference} focused on a fixed-precision setting, while ours is based on the ``Successive Rejects" algorithm introduced in \cite*{gabillon2011multi} focused on a fixed-budget setting. 
Both their work and our work share the spirit of integrating the inner-layer optimization and the outer-layer selection to enhance algorithm performance, but from a different perspective. \cite{wang2024selection} considered a related setting but with different goals compared to our work.




\textbf{Notations.} We denote $[K]$ to be set of $\{1,2,\ldots,K\}$.  We sometimes use $[K]$ as an abbreviation for $\{1,2,\cdots,K\}$ when there is no ambiguity. Let $\lfloor\cdot \rfloor$ be the floor function. And $|\mathcal{A}|$ stands for the cardinality of the set $\mathcal{A}$. $\mathcal{N}(\mu,\sigma^2)$ and  $\mathrm{Poi}(\lambda)$ denote the normal distribution with mean $\mu$ and variance $\sigma^2$ and the Poisson distribution with rate $\lambda$, respectively. Let $\Pi_{\mathcal X}(x)$ be the projection of $x\in \mathbb R^d$ to $\mathcal X$ in the sense of Euclidean norm.

\section{Setting}
Suppose that a decision maker needs to select one from  $K$ systems, labeled as $1,2,\ldots,K$. We denote the \textit{optimizing performance} of the $i$-th system as $v_i$, which is defined as the optimal objective value of an inner-layer optimization. Specifically,
\begin{equation}
v_i = \sup_{x \in \mathcal{X}_i} \E[G_i(x)],
\label{eqn:v}
\end{equation}
in which $\mathcal{X}_i$ is a space (could be a function space) that represents the inner-layer optimization for system $i$, and $G_i(x)$ is a finite-variance random variable that represents the stochastic system performance under decision $x$ for the $i$-th system. The selection of the best system is to select the system with the best optimizing performance, formally given by
\[
\max_{i\in\{1,2,\cdots,K\}} v_i \,=\,\max_{i\in\{1,2,\cdots,K\}} \sup_{x \in \mathcal{X}_i} \E[G_i(x)].
\]
The decision maker has access to choose any $i$ and $x$ and draw a sample of $G_i(x)$.
We discuss two concrete and different settings as follows, which will be the main problem settings for algorithm design and analysis in this work.

\subsection{Simulation Optimization}
The optimizing performance of the $i$-th system, denoted by $v_i$, involves a simulation optimization problem as the inner-layer optimization. Specifically, the inner-layer optimization is given by
\begin{equation}
    v_i = \max_{x \in \mathcal{X}_i} f_i(x) \triangleq \max_{x \in \mathcal{X}_i} \E[F_i(x,\xi)], \label{eq:obj_so}
\end{equation}
in which $x$ denotes the choice of decision variable in a compact and convex set $\mathcal{X}_i$, $\xi$ summarizes all the system randomness, and $F_i(\cdot,\cdot)$ is a deterministic function that captures all the (complicated) system logic and outputs a system performance. The expected performance function $f_i(\cdot)$ is presumed to be continuous so that the maximum can be attained over a compact set. The goal of \textit{Selecting the Best Optimizing System} (SBOS) problem in this simulation optimization setting is to optimize
\[
\max_{i\in\{1,2,\cdots,K\}} v_i \,=\,\max_{i\in\{1,2,\cdots,K\}} \max_{x \in \mathcal{X}_i} \E[F_i(x,\xi)].
\]
We consider settings in which $f(x_i)= \E[F_i(x,\xi)]$ is unknown but can be estimated through expensive simulation samples $F_i(x,\xi)$, when $i$ and $x$ are both specified. For complicated stochastic systems, the most time consuming part often comes from the evaluation of the function $F_i(\cdot,\cdot)$, which summarizes all the complicated system logic and operational rules. In this context, generating one simulation sample refers to one function evaluation of $F_i(x,\xi)$, associated with one gradient evaluation of $\partial_x F_i(x,\xi)$, at a given choice $i$ and $x$. We consider a fixed-budget setting where a budget $T$ is defined as the total number of samples that can be used to generate independent function and gradient calls of $F_i(x,\xi)$'s, adding up over all $i$ and $x$. The budget can be sequentially spent, in the sense that one can decide where to spend the next sample after observing outcomes from all previous samples. After the budget is used up, one needs to decide which system has the best optimizing performance $\max_{i\in\{1,2,\cdots,K\}} v_i$. The goal is to design easy-to-implement algorithms that sequentially allocate simulation samples and eventually achieve provably small probability of false selection (PFS) after the budget $T$ is spent.

\subsection{Data-driven Stochastic Optimization}
The optimizing performance of the $i$-th system, denoted as $v_i$, involves a stochastic optimization problem as the inner-layer optimization, given by
\begin{equation}
v_{i}=\sup_{g\in \mathcal{F}_{i}}\mathbb{E}_{P^{i}}\left[ g(X)\right], \label{eq:data_driven_SO}
\end{equation}
where $\mathcal{F}_i$ can be parametric or non-parametric function classes, and $X$ denotes a general-dimensional random variable having distribution $P^i$. Different from the simulation optimization setting in Section 2.1, here the evaluation of function $g$ is not the bottleneck for data-driven stochastic optimization problems. However, the distributions $P^i$ for $i=1,2,\cdots,K$ are not known and need to be estimated from collecting real-world data samples. This setting notes that each data sample is \textit{costly to collect}, rather than that the computation or function evaluation is costly. Specifically, we consider scenarios where independent and identically distributed (i.i.d.) samples that come from the true unknown distribution $P^i$ for system $i$ can be collected at a cost. A budget $T$ represents the total number of i.i.d. samples that can be collected aggregated for all $K$ systems. The collection of one sample refers to obtaining one i.i.d. observation from the distribution $P^i$ for some $i$. Given the nature of \eqref{eq:data_driven_SO}, no unbiased estimator for $v_i$ is available given finite samples. The goal is to design easy-to-implement algorithms that sequentially decide which sample to collect and eventually decide which system achieves the best optimizing performance $\max_{i\in\{1,2,\cdots,K\}} v_i$.

We have now introduced two classes of SBOS problems - one class on simulation optimization (Section 2.1) and the other class on data-driven stochastic optimization (Section 2.2). In the rest of this work, we will present algorithm design and analysis for the class of simulation optimization problems in Section 3 and present algorithm design and analysis for the class of data-driven stochastic optimization problems in Section 4. We summarize that the key technical difference between these two settings are how the budget is counted and how one sample is defined. Such technical difference captures different sets of applications and demands algorithm design and analysis respectively.

\section{Algorithm and Analysis for Simulation Optimization}
\subsection{Algorithm for SBOS Simulation Optimization Problems}
In this section, we focus on the class of SBOS simulation optimization problems as formulated in Section 2.1. We present our algorithm which is named Sequential Elimination for Optimizing systems (SEO). The SEO algorithm integrates the stochastic gradient descent method in the inner layer and the sequential elimination method in the outer layer. The sequential elimination method is motivated by \cite*{audibert2010best}, \cite{karnin2013almost}, and \cite*{frazier2014fully}. 
For the outer layer, given the number of systems $K$, the basic idea is to divide the budget into $L = \lfloor \log_2(K) \rfloor $ phases. In each phase, roughly speaking, the algorithm evenly allocates the budget to each system that still remains considered. Within each phase, the budget that is allocated to each system is used to solve the inner layer optimization. For the inner layer optimization of a system, the algorithm performs stochastic gradient descent (SGD) using all the allocated budget and then obtains a (biased) estimator of the optimizing performance of that system. At the end of each phase, the algorithm eliminates the bottom half of systems. The elimination is based on the estimated optimizing performance for all the systems under consideration up to that phase. The full procedure of our proposed SEO algorithm is summarized in Algorithm \ref{alg:SE}.

It is evident that, with finite number of samples, the inner layer optimization cannot be completely solved, and the decisions recommended for the inner layer optimizations
are non-optimal. A major challenge for designing and analyzing the SEO algorithm is that we need to balance the bias (compared to the optimal) arising from non-optimal decisions and the variance of each random sample. As a further challenge, unlike theory for standard stochastic optimization problems, we need to estimate the optimal objective value rather than the optimal solution. This is because the comparison between systems is based on their optimal objective function value rather than the optimal choice of decision variable. Therefore, we need to carefully design and analyze the SGD method used in the algorithm and the corresponding estimators.

\bigskip

\begin{breakablealgorithm}
\caption{Sequential Elimination for Optimizing Systems (SEO) in Simulation Optimization}
\label{alg:SE}
\begin{algorithmic}[1]
\State \textbf{Input: } $T$, $K$, the initial choice of decision $x_{1,i}$ for each $i\in[K]$, and a step-size coefficient $\gamma_0$.
\State Set  $L = \lfloor\log_2(K) \rfloor$ and $ \mathcal{A}_1 = \{1,2,\ldots,K\}$.
\For{$\ell = 1,\ldots,L$}
    \State Let $T_\ell \gets \lfloor \frac{T}{L |\mathcal{A}_\ell|}\rfloor$.
    \State set $\gamma =T_\ell^{-1/2} \gamma_0$.
    \For{$i$ in $\mathcal{A}_\ell$}
    \For{$t=1,\ldots, T_\ell$}
    \State Set $x_{1,i}$ to be ending value of the previous iteration of the system $i$: $x_{T_{\ell -1},i}$.
    \State Run a simulation path to collect the system value $F_i(x_{t,i},\xi_{t,i})$  for $i \in \mathcal{A}_\ell$ and the
    \State \,\, associated stochastic gradient at point $x_{t,i}$, $G_i(x_{t,i},\xi_{t,i})$.
    \State Perform $x_{t+1,i} \gets \Pi_{\mathcal X_i}( x_{t,i}+\gamma G_i(x_{t,i},\xi_{t,i})).$
    \EndFor
    \State Compute $\hat{v}_i^{T_\ell} \gets\frac{1}{T_\ell} \sum_{t=1} ^{T_\ell} F_i(x_{t,i},\xi_{t,i})$.
     \EndFor
     \State For all $i\in\mathcal{A}_\ell$, sort the systems by $\hat{v}_i^{T_\ell} $ in a descending order.
   \State Let $\mathcal{A}_{\ell+1}$ contain the top $\lfloor |\mathcal{A}_\ell|/2 \rfloor$ systems in $\mathcal{A}_\ell$.
\EndFor
\State \textbf{Output: } the remaining system in $\mathcal{A}_{L+1}$.
\end{algorithmic}
\end{breakablealgorithm}
\subsection{Algorithm Performance Guarantee: Theory and Analysis}
\label{sec:theory:sim}
In this subsection, we prove a performance guarantee for the SEO algorithm (Algorithm \ref{alg:SE}) that is designed for SBOS problems in the simulation optimization setting. A key obstacle in the analysis is to bound the bias in the estimator for the optimal objective value in the inner-layer optimization and to control how the bias from the inner-layer optimization affects the outer-layer selection. When analyzing the bias, a major challenge arises because the algorithm needs to average out all the samples including those which may be farther from the optimal value to reduce the variance. Before presenting the analysis and theory, we first state the assumptions.


\begin{assumption} Let $-G_{i}(x,\xi )\in \partial _{x}(-F_{i}(x,\xi ))$ be a subgradient. Let $\mathcal X_i$ be {bounded and closed convex domains}. Denote
	$D_{\mathcal{X}_i}=\frac{1}{2}\max_{x,x^{\prime }\in \mathcal{X}_i}\left\Vert
	x-x^{\prime }\right\Vert _{2}^{2}.$
\label{assump:gd} For each $i\in[K]$, we assume that $F_{i}(x,\xi )$ satisfies the following
assumptions:

\begin{enumerate}
\item $F_{i}(x,\xi)$ is concave in $x$ and finite-valued for any $\xi$.

\item The probability distribution of $F_{i}(x,\xi )$  has regularized tails, given by
\begin{equation*}
\mathbb{E}[\exp \left( \lambda \left( F_{i}(x,\xi )-f_{i}(x)\right) \right)
]\leq \exp \left( \frac{\lambda ^{2}\sigma _{F,i}^{2}}{2}\right) ,\text{ }%
\forall \lambda \in \mathbb{R}\text{,}
\end{equation*}
where $\sigma_{F,i}$ are positive real numbers that can depend on $i$.

\item Let $-f_i'(x)=\partial (-f_i(x))$ be a subgradient for the expected system performance function $f_i(x)$. The subgradient estimator $-G_{i}(x,\xi )$ is unbiased in the sense $\mathbb{E}[G_{i}(x,\xi )]=f_{i}^{\prime }(x)$ for all $x$. Also, the variances and tail conditions for the subgradient estimator are regularized as  $\mathbb{E}[\left\Vert
G_{i}(x,\xi )-f_{i}^{\prime }(x)\right\Vert _{2}^{2}]\leq \sigma _{G,i}^{2}$
and
$\mathbb{E}\left[ \exp \left( \left\Vert G_{i}(x,\xi )-f_{i}^{\prime
}(x)\right\Vert _{2}^{2}/\sigma _{G,i}^{2}\right) \right] \leq \exp (1),$
where $\sigma_{G,i}$ are positive real numbers that can depend on $i$.

\item There exists $M>0$ such that $\left\Vert f_{i}^{\prime }(x)\right\Vert _{2}\leq M$ for all $x$ and $i$.
\end{enumerate}
\label{Assump:sim}
\end{assumption}
The concavity in Assumption \ref{Assump:sim}.1 is needed to prove the convergence to the global maximum. Otherwise, the algorithm can (and need to) be modified to have multiple random initializing points. Assumptions \ref{Assump:sim}.2 and \ref{Assump:sim}.3  regularize the tail conditions of the random objective function $F_i(x,\xi)$ and the stochastic gradient $G_i(x,\xi)$, where for example Gaussian distribution assumption would be a special case. Assumption \ref{Assump:sim}.4 is equivalent to the Lipschitzness condition for $f_i(x)$. We note that these assumptions are typically needed in the continuous stochastic optimization literature that establishes convergence rates. Then, we have the following convergence result from $\hat{v}_i^T$ to $v_i$, where $v_i$ is the optimizing performance for system $i$ as defined in (\ref{eq:obj_so}) and $\hat{v}_i^T$ denotes the estimator for $v_i$ after $T$ steps of SGD (where $T$ is a dummy variable), as shown in line 12 of Algorithm \ref{alg:SE}.

\begin{proposition}
	\label{prop:conv}
Suppose Assumption \ref{assump:gd} is enforced. For the constant-step size
policy, where
$\gamma =\frac{D_{\mathcal{X}_i}}{\sqrt{T(M^{2}+\sigma _{G,i}^{2})}},$
we have for any $\epsilon >0$ the following holds%
$$
\mathbb{P} \left( \left\vert \hat{v}_{i}^{T}-v_{i}\right\vert \geq \frac{%
3D_{\mathcal{X}_i}\sqrt{M^{2}+\sigma _{G,i}^{2}}}{\sqrt{T}}+\epsilon \right)
\leq 3\exp \left( -\frac{T\epsilon ^{2}}{\left( 3\sqrt{3}\sigma _{G,i}D_{%
\mathcal{X}_i}+\frac{\sigma _{F,i}}{\sqrt{2}}\right) ^{2}}\right) +\exp \left( -\frac{%
\epsilon T\sqrt{T}}{3\sigma _{G,i}D_{\mathcal{X}_i}+\frac{\sigma _{F,i}}{\sqrt{6}}}%
\right) .
$$
\end{proposition}

\textbf{Remark:} The algorithm depends on the choice of the step size $\gamma$, which involve a few parameters to select. The establishment of the high probability convergence theory results as shown above would then raise a discussion on how the choice of parameters affect the results. We note that this is an open question and discussion for the standard stochastic convex optimization. 


\textbf{Remark:} Although in Algorithm 1, we use stochastic gradient methods for the inner optimization (lines 7-12), our Theorem 1 applies to other generic maximization methods with the same convergence rate.

The detailed proof is in \ref{ec:proof_sec_2}. Proposition \ref{prop:conv} shows that if the SGD scheme is chosen appropriately, the estimated objective value converges to the true optimal objective value exponentially fast as the sampling size grows to infinity, even in the presence of bias. Proposition \ref{prop:conv} controls the bias rate in the estimated optimal objective value, which to our knowledge, is an independent contribution, given that the literature largely focuses on the optimizer property instead of the objective value. By utilizing Proposition \ref{prop:conv}, we have the following bound for the probability of false selection of Algorithm \ref{alg:SE}. Note that the probability of false selection (PFS) is given by $\mathbb{P}(1\notin \mathcal{A}_{L+1})$ where $\mathcal{A}_{L+1}$ is the set returned from Algorithm \ref{alg:SE} that contains only one system. 
\begin{theorem}
Suppose Assumption \ref{assump:gd} is enforced and $v_1>v_2\geq v_3 \geq \ldots \geq v_K$. Let $\Delta_i = v_1-v_i$ for $i = 2,3,\ldots,K$. When
$T\geq \lfloor\log _{2}\left( K\right)\rfloor K\left( \max_{i\in \{2,3,\ldots,K\}}\frac{12D_{\mathcal{X}_i}}{\Delta
_{i}}\left( \sqrt{M^{2}+\sigma _{G,i}^{2}}\right) \right)
^{2},$
we have the output from Algorithm \ref{alg:SE} satisfy
\begin{align}
&\mathrm{PFS} = \mathbb{P}(1\notin \mathcal{A}_{L+1})\notag \\
\leq &\lfloor \log_2(K) \rfloor \left\lbrace 24\exp \left( -\frac{ T }{192 \lfloor\log_2(K)\rfloor M_\sigma ^{2}H_{2}(v)}\right) +8\exp \left( -\frac{%
	T\sqrt{T}}{16  \sqrt{K \lfloor \log_2(K) \rfloor^2} M_\sigma H'_{2}(v)} \right)\right\rbrace  , \label{eq:bd:sim}
\end{align}%
where $H_{2}(v)=\max_{i>1}\frac{i}{\Delta _{i}^{2}}$, $H_{2}^{\prime
}(v)=\max_{i>1}\frac{i}{\Delta _{i}}$, and $M_\sigma=
\max_{i\in[K]}\{3\sigma _{G,i}D_{%
\mathcal{X}_i}+\sigma _{F,i}/\sqrt{6}\}$.
\label{thm:sim}
\end{theorem}

The detailed proof is in \ref{ec:proof_sec_2}.  This result also includes \cite*{gabillon2011multi} and \cite*{carpentier2016tight} as special cases, which do not have inner-layer optimizations in each system. The bound (\ref{eq:bd:sim}) can be dominated by the first term in the right hand when $T$ is much larger than $K\log_2(K)$.  Further, we observe that the bound is exponentially linear on $T$. And the rate is exponentially inverse proportional to the log of the number of systems $\log_2(K)$, the complexity term $H_2(K)$, and the variance term $M$. Based on Theorem 1, we have demonstrated a reliable performance guarantee for the SEO algorithm that have desirable dependence on the budget $T$ (exponential decay) and on the number of systems $K$. Note that, within each system, there are technically a infinite continuum of ``sub-systems", a challenge that is overcome by the SEO algorithm by exploring the concavity structure.

Following the upper bound exponential rate result to control the PFS, we also provide a brief lower bound result in Proposition \ref{thm:ld:queue}, utilizing the results from \cite*{carpentier2016tight}. Here, we first define the oracle model, which is similar to the setting discussed in \cite*{agarwal2009information} and \cite*{nemirovskij1983problem}. $K,T$ and $\mathcal{X}_i$ are preknown to the decision maker. At time $t \in [T]$, the decision maker chooses a system $i \in [K]$ and also queries a point $x \in \mathcal{X}_i$. An oracle answers  the query by giving  $F_i(x,\xi)$ and $G_i(x,\xi)$ . We let $\mathcal{O}_{\bar{\sigma},a}$ to denote the class of all oracles satisfying Assumption \ref{assump:gd} with $\sigma_{F,i}\leq\bar{\sigma}$, $\sigma_{G,i}\leq\bar{\sigma}$ and the complexity term $H(v)\leq a$, where $H(v)=\sum_{i=2}^K (v_1-v_i)^{-2}.$ Then, we have the following lower bound:
\begin{proposition}
\label{thm:ld:queue}
Let $K>1$ and $a>0$. If $T\geq 16 \bar{\sigma}^4 a^2 (4\log(6TK)) / (60)^2$ and $a\geq \tfrac{11}{4}\bar{\sigma}^{-2}  K^2$, then for any algorithm it holds that the algorithm's recommended system by the end of $T$, labeled as $i$, satisfies that
\begin{equation}
\sup_{\mathcal{O}_{\bar{\sigma},a}}\left[  \mathbb{P}\left( i \neq \arg\max_j v_j \right) \times \exp\left(100\frac{T}{\bar{\sigma}^2 \log(K) H(v)} \right) \right] \geq 1/6. 
\end{equation}
\end{proposition} 
\textbf{Remark:} The relation $H_2 \leq H\leq \log(2K) H_2$ holds \citep{audibert2010best}.  Then, Proposition \ref{thm:ld:queue} together with Theorem \ref{thm:sim} shows that the hardest problems are those $H_2$ is of same order as $H$.

Proposition \ref{thm:ld:queue} shows that our upper bounds are tight for the complexity term $H_2$ and the variance term $\sigma_{F,i}$ up to constant.
\section{Algorithm and Analysis for Data-driven Stochastic Optimization}
\subsection{Algorithm for SBOS Simulation Optimization Problems}
In this section, we present the Sequential Elimination for Optimized Systems (SEO) algorithm designed for SBOS problems in the data-driven stochastic optimization setting, as introduced in Section 2.2. In this setting, the bottleneck in terms of cost is not the simulation evaluation cost, but is the number of real data samples we can collect. We presume that the function evaluation cost of $g(\cdot)$ is much cheaper compared to the cost of collecting real data. Therefore, the sampling budget only counts the number of collected real data samples. Specifically, we assume there is an oracle that effectively solves the following sample average approximation problem
\begin{equation}
	\hat{v}_{i}^{{n_{i}}}=\sup_{g\in \mathcal{F}_{i}}\mathbb{E}_{P_{n_{i}}^{i}}\left[ g(X)%
	\right],  \label{prob:saa}
\end{equation}%
where $P_{n_{i}}^{i}$ denotes the empirical distribution with $n_{i}$ data
samples from $P^i$. Algorithm \ref{alg:SE2} details our method. Intuitively, for the outer layer, the algorithm performs sequential elimination. In the inner layer, the algorithm draws the oracle to solve the sample average approximation problem (\ref{prob:saa}). We note that despite of the algorithm's simple form, which itself is an advantage, the performance guarantee analysis for the algorithm remains challenging.

\bigskip

\begin{breakablealgorithm}
\caption{Sequential Elimination for Optimizing Systems (SEO) in Data-driven Stochastic Optimization}
\label{alg:SE2}
\begin{algorithmic}[1]
\State \textbf{Input: } $T$ and $K$.
\State Set  $L = \lfloor \log_2(K) \rfloor$ and $ \mathcal{A}_1 = \{1,2,\ldots,K\}$.
\For{$\ell = 1,\ldots,L$}
    \State Let $T_\ell\gets \lfloor \frac{T}{L |\mathcal{A}_\ell|}\rfloor$.
   \State Collect $T_\ell$ samples for each system in $\mathcal{A}_\ell$.
   \State For each $i \in \mathcal{A}_\ell$, solve the SAA problem (\ref{prob:saa}) using the collected samples for that system and obtain an estimation $\hat{v}_{i}^{{n_{i}}}$ for $v_i$.
   \State Let $\mathcal{A}_{\ell+1}$ contain the top $\lfloor |\mathcal{A}_\ell|/2 \rfloor$ systems in $\mathcal{A}_\ell$ ordered by $\hat{v}_{i}^{{n_{i}}}$ for $i \in \mathcal{A}_\ell$.
\EndFor
\State \textbf{Output: } the remaining system in $\mathcal{A}_{L+1}$.
\end{algorithmic}
\end{breakablealgorithm}

\smallskip

\subsection{Performance Guarantee in Data-driven Stochastic Optimization}
In this subsection, we prove performance guarantee for the proposed SEO algorithm to solve SBOS problems in the data-driven stochastic optimization setting. In order to quantify the favorable biasing caused by overfitting, we need a complexity notion of the function classes $\mathcal{F}_i$. We first define the covering number \citep[Definition 5.1]{wainwright2019high}.
\begin{definition}[Covering number]
	A $\delta$-cover of a set $\mathcal{F}$ with respect to a metric $\rho$, $N(\delta,\mathcal{F},\rho)$ is a set $\{g_1,\dots,g_N\}\subset \mathcal{F}$ such that for each $g\in \mathcal{F}$, there exists $i\in[N]$ that $\rho(g,g_i)\leq \delta$.	
	The $\delta$-covering number $N(\delta,\mathcal{F},\rho)$ is the cardinality of the smallest $\delta$-cover.
	\end{definition}
	
Then, the complexity of the set $\mathcal{F}$ is measured by the entropy integral \citep[(5.45)]{wainwright2019high} defined below.
\begin{definition}[Entropy integral]
Define
	\begin{equation}
	\mathcal{J}(\mathcal{F},P) \triangleq  \mathbb{E}_{P^{\otimes n} }\left[ \int_{0}^{+\infty}\sqrt{\log N(t,\mathcal{F},\left\Vert \cdot
		\right\Vert _{P _{n}})}\mathrm{d}t\right] ,
	\end{equation}
where $P _{n}$ denotes the $n$-times product measure of $P$, $P_n$ is the empirical distribution with $n$ i.i.d. samples from $P$, and the metric $\left\Vert \cdot
\right\Vert _{P _{n}}$ is defined by
$\left\Vert \cdot \right\Vert _{P_{n}}\triangleq \sqrt{\frac{1}{n}%
		\sum_{i=1}^{n}\left( f(z_{i}\right) -g(z_{i}))^{2}}.$
	\end{definition}
	
 We provide several instances that have known and finite entropy integral below.

\begin{example} The following functional classes have finite entropy integrals.
	\begin{itemize}
		\item \textbf{Vector spaces \citep[Example 19.16]{van2000asymptotic}:} Let $\mathcal{F}$ be the set of all linear combinations $\sum \lambda_i f_i$ of a given, finite set of functions $\{f_1,\ldots,f_k\}$ on $\mathcal{X}$. Suppose $\mathcal{F}$ is uniformly bounded in $\mathcal{X}$. Then,  $\mathcal{F}$ has finite entropy integrals.
		\item \textbf{Lipschitz parametrized class:} Suppose that $\mathcal{F} = \{g(\theta,\cdot):\theta \in \Theta\}$ is a parametrized class, where $\Theta$ is a $d$-dimensional unit Euclidean ball $B_2^d \subset \mathbb{R}^d$. And we assume for all $x$, $|g(\theta,x)-g(\theta',x)|\leq L\|\theta-\theta'\|_2$. Then, $\mathcal{J}(\mathcal{F},P) = O(L\sqrt{d})$. The proof follows the covering number bound in \cite[Example 5.18]{wainwright2019high}.
		\item \textbf{VC classes \citep[Example 5.24]{wainwright2019high}:} Let $\mathcal{F}$ be a class of $\{0,1\}$-valued functions with VC-dimension $d$ \citep{vapnik1971uniform}, then $\mathcal{J}(\mathcal{F},P) = O(\sqrt{d})$.
	\end{itemize}
	\end{example}

\smallskip	

More examples can be found in \citet[Chapter 19]{van2000asymptotic} and \citet[Chapters 4\&5]{van2000asymptotic}.

To help our analysis of the probability of the false selection, we assume the  function classes are uniformly bounded and have finite entropy integrals.
\begin{assumption} The function classes $\{\mathcal{F}_i\}_{i=1}^{K}$  satisfy:
	\begin{itemize}
	\item There exits $B>0$ such that $f(z)\in \lbrack -B,B]$ for all $f\in \bigcup_{i=1}^{K}\mathcal{F}%
_{i}.$
\item $\max_{i \in [K]} \mathcal{J}(\mathcal{F}_i,P^i)<+\infty$.
\end{itemize}
\label{assump:bd+finite}
\end{assumption}
We emphasize here that $P^i$ does not depend on $g$.
Then, we are ready to show our results on the upper bound of the probability of false selection (PFS) for Algorithm \ref{alg:SE2}.

\begin{theorem}
	Suppose Assumption \ref{assump:bd+finite} is enforced and $v_1>v_2\geq v_3 \geq \ldots \geq v_K$. Let $\Delta_i = v_1-v_i$ for $i = 2,3,\ldots,K$. When
		$T\geq \lfloor\log _{2}\left( K\right)\rfloor K\left(\max_{i\in [ K]} \frac{192}{\Delta
			_{i}}\left( \mathcal{J}(\mathcal{F}_i,P^i)\right) \right)
		^{2}$, we have the output $\mathcal{A}_\ell$ from Algorithm \ref{alg:SE2} satisfies
	\begin{equation*}
		\mathrm{PFS} = \mathbb{P}(1\notin \mathcal{A}_{L+1})\leq 6\lfloor\log _{2}(K)\rfloor\exp \left( -\frac{T}{%
			128B^{2}\lfloor\log _{2}\left( K\right)\rfloor H_{2}(v)}\right) ,
	\end{equation*}%
	where $H_{2}(v)=\max_{i>1}\frac{i}{\Delta _{i}^{2}}.$
	\label{thm:SA}
\end{theorem}
\label{sec:theory:SA}

The detailed proof is in \ref{ec:proof:SA}. The biases from overfitting the collected data are controlled by the complexity of the function class $\mathcal{J}(\mathcal{F}_i,P^i)$. Theorem \ref{thm:SA} provides an exponential rate of convergence of PFS as the budget increases to infinity, establishing a performance guarantee for the SEO algorithm applied to SBOS problems in the data-driven stochastic optimization setting. 

Similar to Proposition \ref{thm:ld:queue}, we provide a lower bound in Proposition \ref{thm:ld:ssa}.
We let $\mathcal{O}_{B,a}$ denote the class of all oracles satisfying Assumption \ref{assump:bd+finite} and the complexity term $H(v)\leq a$. Then, we have the following lower bound:
\begin{proposition}
\label{thm:ld:ssa}
Let $K>1$ and $a>0$. If $T\geq 16 B^4 a^2 (4\log(6TK)) / (60)^2$ and $a\geq \tfrac{11}{4}B^{-2}  K^2$, then for any algorithm that return systems $i^*$ at time $T$, it holds that
\begin{equation}
\sup_{\mathcal{O}_{B,a}}\left[  \mathbb{P}\left( i \neq \arg\max_i v_i \right) \times \exp\left(100\frac{T}{B^2 \log(K) H(v)} \right) \right] \geq 1/6.
\end{equation}
\end{proposition} 
The proof of Proposition \ref{thm:ld:ssa} follows the same routine as the proof of Proposition \ref{thm:ld:queue}.
%


\section{Applications}
\label{sec:app}
In this section, we present three applications that need the selection of the best optimizing system. Two applications correspond to the setting in Section 3, and one application corresponds to the setting in Section 4. For each application, we describe the problem setting, implement our proposed algorithm, and compare it with the uniform sampling algorithm. For the simulation optimization (Section \ref{sec:sim:app}) and the selection of the best drug (Section \ref{sec:drug:app}) applications, we further compare our algorithm with the Optimal Computing Budget Allocation (OCBA) algorithm \citep{chen2000simulation} with discretization. The uniform sampling algorithm is that we treat each system in a uniform way by allocating a load $T/K$ samples to each system. Each system receives the same number of samples to solve the inner-layer optimization, using the same approach as in our proposed SEO algorithm. For the OCBA algorithm, we adopt the variant proposed in \cite*{chen2000simulation} and \cite*{wu2018analyzing}  with the size of samples for an initial estimation $N_0$ linear in $T$, since \cite*{wu2018analyzing} shows that a fixed $N_0$ will not result in an exponential convergence rate. The details are listed in Algorithm \ref{alg:OCBA} in \ref{EC.1:algo}. We show that our proposed algorithm consistently outperforms the two benchmarks regarding the probability of false selection, for different numbers of systems $K$ and different total budget $T$.

\bigskip

\subsection{Optimal staffing and pricing in queueing simulation optimization}
\label{sec:sim:app}
In this example, we apply our proposed method to a simulation optimization problem in the queueing context, with the goal of selecting the best staffing plan for a two-station service system under optimized pricing plans. Specifically, we consider a first-in-first-out service system with two connected stations, Station One and Station Two. The service system has in total $K$ homogeneous staff members (servers). The system manager needs to select $x\in\{1,2,\cdots,K-1\}$ staff members to serve at Station One and $K-x$ staff members to serve at Station Two. Each station has a first-in-first-out logic with infinite waiting room capacity. Station One offers a type-one service, and Station Two offers a type-two service. The type-one service is required to be completed before type-two service. That is, customers who enter the system always first join Station One to receive type-one service. Upon completion of service in Station One, customers will immediately join Station Two to receive type-two service. The specifics are given as follows.

\textbf{Arrival process.} The system is open to arriving customers on $[0,H]$. The arrival process of customers to the system is a non-stationary Poisson process with time varying rate $\{\lambda(t):t\in[0,H]\}$. Consider $\lambda(t) = {\lambda}_0\cdot  t(H-t) /H^2$. The system runs until the last customer completes services.
\textbf{Service times.} For the $i$-th customer, the type-one service time requirement $S_{i,1}$ and type-two service time requirement $S_{i,2}$ are jointly distributed log-normal distributions, with parameters $\mu_1,\mu_2,\sigma_1^2,\sigma_2^2,\rho$. Specifically, let $Y_{i,1}$ and $Y_{i,2}$ be jointly distributed Gaussian random variables with mean vector $(\mu_1,\mu_2)$ and covariance matrix $((\sigma_1^2,\rho \sigma_1\sigma_2),(\rho\sigma_1\sigma_2,\sigma_2^2))$. Then $(S_{i,1},S_{i,2})$ has the same distribution as $(\exp(Y_{i,1}),\exp(Y_{i,2}))$. The pairs $(S_{i,1},S_{i,2})$ for $i=1,2,3,\cdots$ are independent and identically distributed. \textbf{Abandonment and patience.} The $i$-th customer has a patience time $Pa_i$, independently and identically distributed according to an gamma distribution with rate parameter $\beta_a$ and shape parameter $\alpha_a$. The customer abandons the system when and only when she waits for more than $Pa_i$ time in the waiting room of Station One. \textbf{Pricing of service and customer reaction.} The system can set a price $p\in[0,1]$ and there is an elasticity function $q(p)=1-p$ for customers. That is, if the price is set as $p$, then each arriving customer from the aforementioned non-stationary Poisson process  has an independent probability $q(p)$ of accepting the price and entering the system, but otherwise rejecting the price and immediately leaving the system. \textbf{Queueing performance and objectives.} For each given staffing plan $\{x,K-x\}$ and service price $p$, denote $D(x,p,H)$ as total number of customers that end up accepting the price and receiving services in the system. The system receives an overall reward $pD(x,p,H)$. Denote $W(x,p,H)$ as the total amount of waiting times for all customers in either station. The system receives an overall penalty $c W(x,p,H)$ that is proportional to the total amount of waiting time.

\textbf{Optimization Goal.} The goal is to select the best staffing plan that maximizes the expected net reward. For each staffing plan, the price needs to be optimally set to maximize the expected net reward associated with that plan. Specifically, the optimization problem is give by
$$\max_{x\in\{1,2,\cdots,K-1\}} f(x) \triangleq \sup_{p\in[0,1]} \E [pD(x,p,H)-cW(x,p,H)]. $$
Note that the most costly computational part is for any given $x$ and $p$ to obtain a sample of  $pD(x,p,H)-cW(x,p,H)$, which requires running through the entire time horizon of system logic. This optimization problem can be classified as a staffing-pricing joint decision making problem, which have been widely considered in the literature and related applications. See \cite*{kim2018value}, \cite*{lee2019pricing}, \cite*{chen2020online} and references within. Most of work in this literature presumes the system to have a steady-state behavior and uses the steady-state vehicle to derive insightful decisions. We alternatively focus on providing a computational tool when some applications desire the selection of an optimized staffing plan but observe non-stationarities and potentially complicated system uncertainties. In presence of non-stationarities and potentially complicated system uncertainties, it is often difficult to derive closed-form solutions and demands the use of Monte Carlo simulation to solve the associated optimization problem. In this example, we do not consider the use of common random numbers, which can be potentially added as an additional tool to improve efficiency for all algorithms in comparison.

For the experiment specifics, we choose $\lambda_0 = 0,\sigma_1 = 1,\sigma_2 =1, \rho =0.5, \mu_1 = \log(10) + \log(K), \mu_2 = \log(2)  + \log(K) $.  Therefore, there is an average $\int_{t=0}^H \lambda(t) \mathrm{d}t =\lambda_0 H/6=1000/3 $ and $\mathbb{E}[S_{i,1}]=10.5 + K$ and $\mathbb{E}[S_{i,2}]=2.5 + K$.   For the patience time distribution of the customer, we set $\beta_a = 1$ and $\alpha_a = 2\mu_1$. In the SEO and uniform sampling algorithm, the step-size constant $\gamma_0$ is chosen as $2/H = 1/1000$ and the initial point $p_1 = 0.5$.  The stochastic gradient $G_i(p,\xi)$ is obtained by finite difference gradient estimator. Specifically, we approximate the gradient by  $(F_i(p,\xi) - F_i(p-\Delta_p,\xi))/\Delta_p$, where we choose $\Delta_p = 0.03$ here. Therefore, to obtain one sample and one gradient, we need to evaluate the function $F_i(\cdot,\cdot)$ twice. To make a fair comparison, we use $T/2$ as the input in Algorithm \ref{alg:SE} for the total budget $T$ shown in $x$-axis of the figures below. For the OCBA approach, we discretize the space $[0,1]$ to 10 possible systems $\{0.1,0.2,\ldots,1\}$.

 Figure \ref{fig:queueing} shows the comparison between SEO, uniform sampling and OCBA algorithms in the optimal staffing and pricing problem. Figures \ref{K=16:q} - \ref{K=128:q} plot the probability of correct selection averaged over $1000$ replications as a function of increasing budget, for $K=16,40,128$, respectively. The black solid line, the blue dashed line, and the orange dotted line represent the SEO, uniform, and OCBA algorithms, respectively. In this experiment, since there does not exist an analytical optimal solution, we do not report the optimality gaps in this setting.

 Here $\xi_j$ are the iid random variables we use to simulate the outcome $pD(x,p,H) - cW(x,p,H)$.
 It is evident to see that our proposed SEO algorithm performs better than the uniform and the OCBA algorithms for almost every $K$ and $T$. The only exception is that when $K=16$ and $T$ small, likely due to the initial estimation bias. Another thing worth noting is that in theory, each line in those figures should be monotonically increasing. The zig-zag phenomena in those figures are due to random errors. \ref{ec:numerical:queue} contain additional plots.

\begin{figure}[htbh]
\centering
\subfigure[$K=16$]{
\label{K=16:q} \includegraphics[width=1.91in]{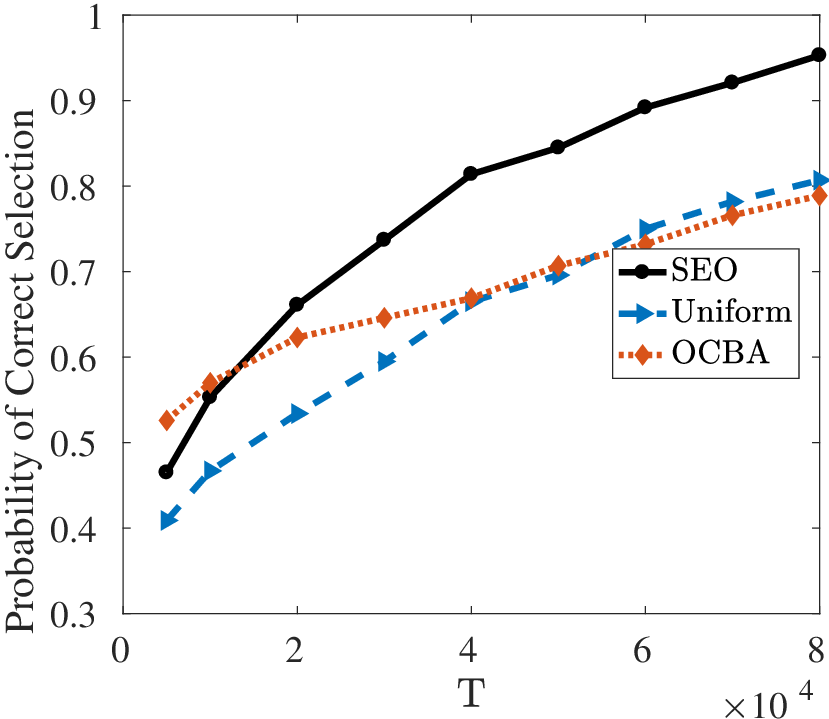}}
\subfigure[$K=40$]{
\label{K=40:q} \includegraphics[width=1.91in]{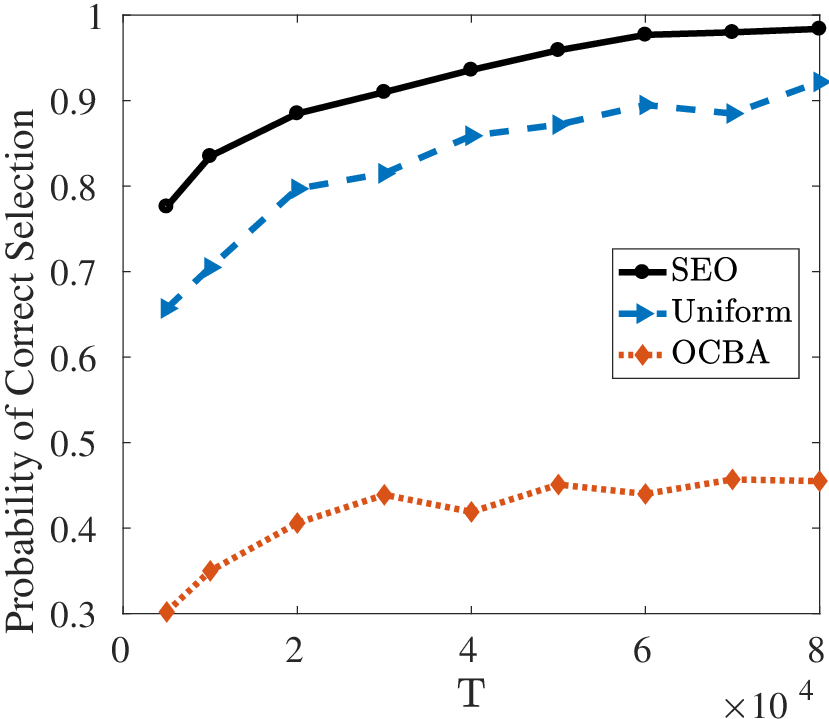}}
\subfigure[$K=128$]{
\label{K=128:q} \includegraphics[width=1.91in]{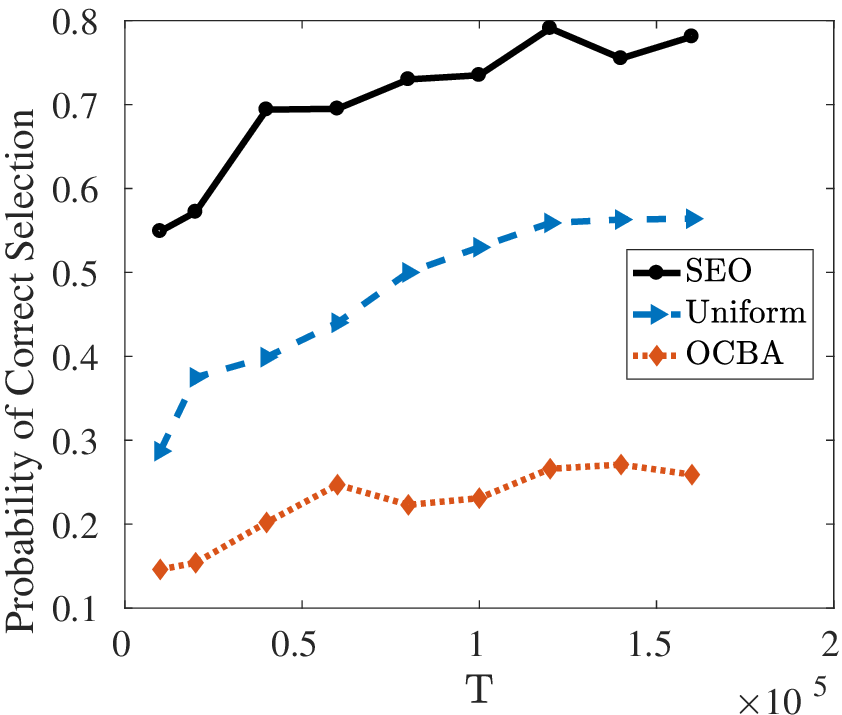}}
\caption{The comparison between SEO, uniform sampling and OCBA in the optimal staffing and pricing problem.}
\label{fig:queueing}
\end{figure}
\label{sec:numerical:sim}
\subsection{Optimal dosage in the selection of the best drug}
In this example, we consider $K$ different drugs (or treatment plans) that are being compared to treat a disease for a targeted population. Each drug can have different expected effect on the population with different dosage amount (\cite{erman2006efficacy,verweij2020randomized}). It is a priori not known for each drug what is the dosage amount that has the best expected effect for that drug among a continuous range of allowable dosage amount. Suppose that one can sequentially do $T$ experiments, where each experiment selects one of the $K$ drugs and a specific dosage of that drug. Suppose that for each experiment, a noisy observation can be obtained on the effect without much delay.
The goal is to select one drug with the best expected effect under the best dosage amount for each drug. In this experiment setting, we presume that for each drug, the expected effect as a function of the dosage amount is concave. This concavity assumption on one hand has been captured by empirical evidence for some drugs (e.g., \cite{verweij2020randomized} identifies a quadratic function form) and on the other hand captures the intuition that neither too small dosage nor too large dosage is desirable.

In the new experiments, we based on the results in \cite{verweij2020randomized}, which  examined the dose-response of aprocitentan. The effect is measured by the mean change from baseline in sitting diastolic blood pressure (SiSBP) and the dosage amount ranges from 0 to 50 mg. The small SiSBP is, the better. Since the data is not public, we fit \citet[Figure 3A]{verweij2020randomized}
as a quadratic function $a_* q^2 +b_* q +c_*$ with $a_*=9/1250, b_* =-23/50,c_* = -5 $ which is also plotted in Figure \ref{fig:dosage}. We call it the ``center" system.
\begin{figure}[!ht]
    \centering
    \includegraphics[width=1.91in]{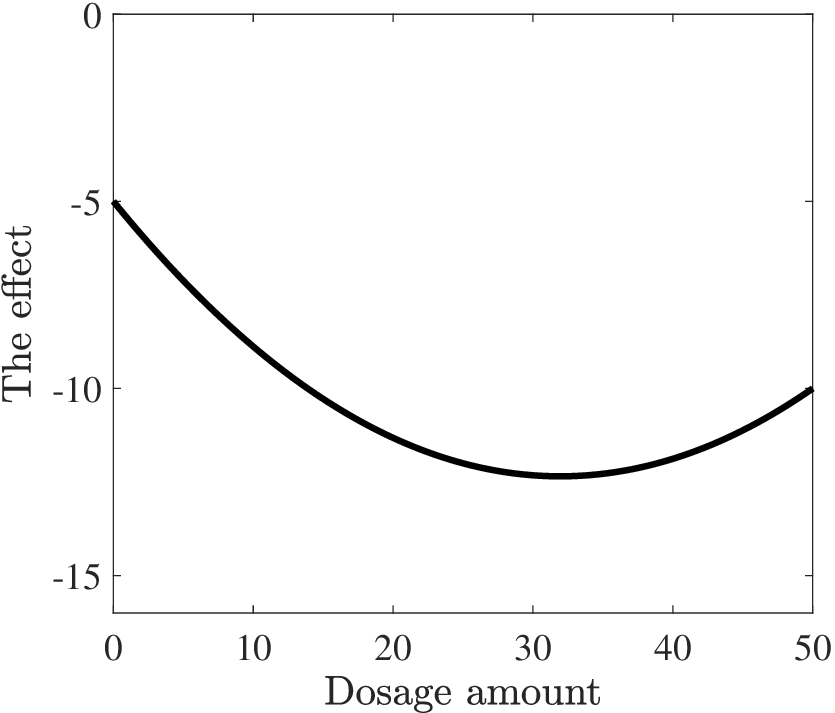}
    \caption{The effect curve with respect to the dosage amount \citep[Figure 3A]{verweij2020randomized}}
    \label{fig:dosage}
\end{figure}
We perturb  the ``center" system to generate $K$ possible systems. Specifically, we first generate $K$ uniform random numbers $u_1,u_2,\ldots,u_K$ supported in $[-0.1,0.1]$. Then, the $i$-th system is a quadratic function of the form $a_i q^2 +b_i q +c_i + \epsilon$ with $[a_i,b_i,c_i] = (1+u_i)\times[a_*, b_* ,c_* ]$ and $\epsilon \sim \mathcal{N}(0,1)$ for $i=1,2,\ldots,K$. For our algorithm (SEO) and  the uniform sampling algorithm, we pick the starting point $x_0=25$ and the step-size constant $\gamma_0=1$.  Similar with the simulation optimization discussed in Section \ref{sec:numerical:sim}, we use finite difference gradient estimator with $\Delta_x=0.5$. The difference is that we cannot use common random number to generate two samples with $x$ and $x-\Delta_x$. Therefore the variance of the gradient will be enlarged. For the OCBA algorithm, we discretize the dosage space as $[11,12,\ldots,40]$.


Figure \ref{fig:dosage_comp} shows the comparison between SEO, uniform sampling and OCBA algorithms in the optimal dosage problem, which is an analog of  \ref{fig:queueing} in Section \ref{sec:numerical:sim}. We see a clear advantage of our algorithm over the other two algorithms. Since the problem is inherently hard problem as we enforce different drugs have similar effects, the probability of correct selection is still high when $K\leq 40$. More plots are contained in \ref{ec:numerical:dose}.

\begin{figure}[htbh]
\centering
\subfigure[$K=16$]{
\label{K=16:dose} \includegraphics[width=1.91in]{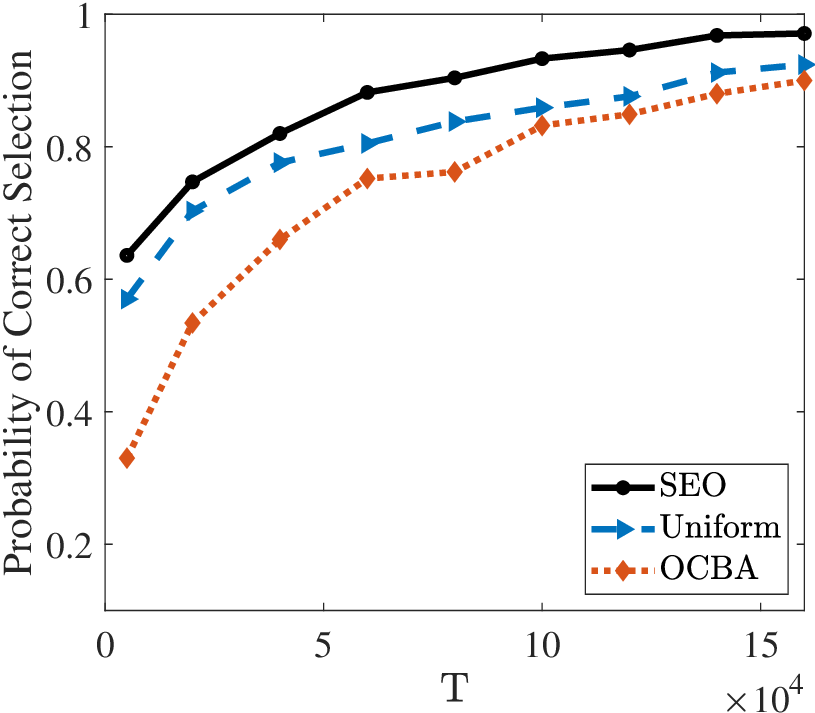}}
\subfigure[$K=40$]{
\label{K=40:dose} \includegraphics[width=1.91in]{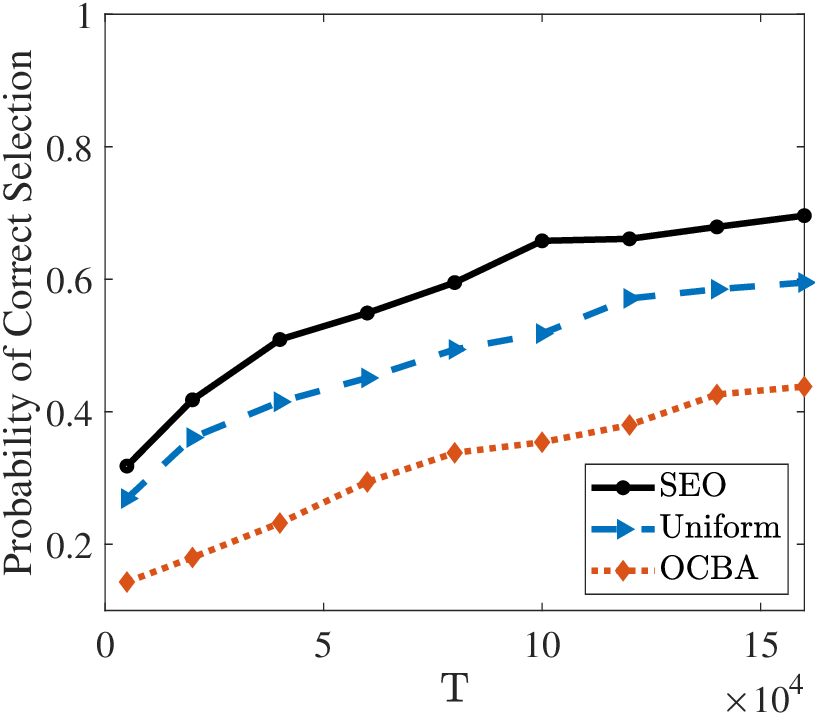}}
\subfigure[$K=128$]{
\label{K=128:dose} \includegraphics[width=1.91in]{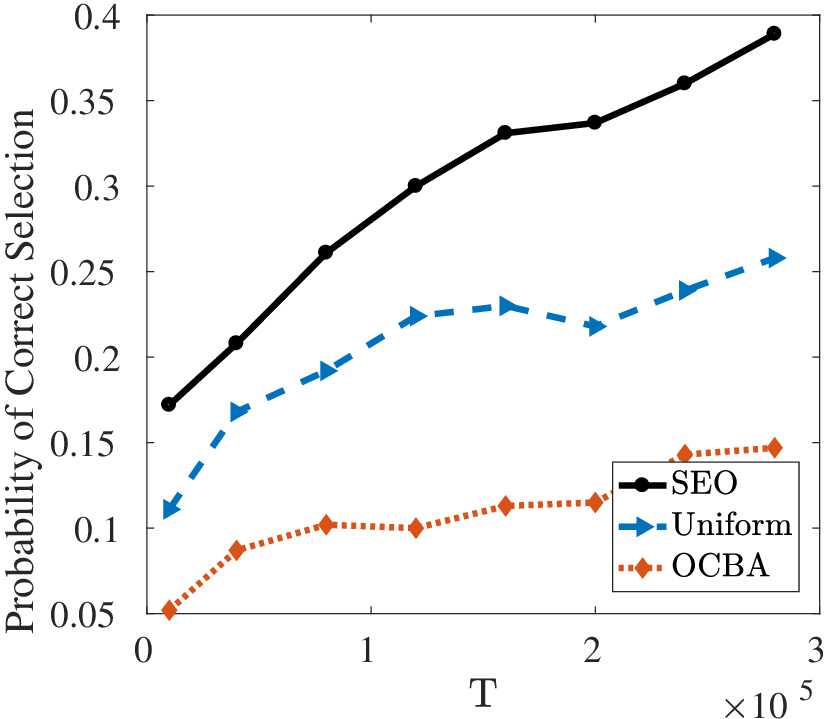}}
\caption{The comparison between SEO, uniform sampling and OCBA in the optimal dosage problem.}
\label{fig:dosage_comp}
\end{figure}
\label{sec:drug:app}

\begin{figure}[htbh]
\centering
\subfigure[$K=16$]{
\label{gK=16:dose} \includegraphics[width=1.91in]{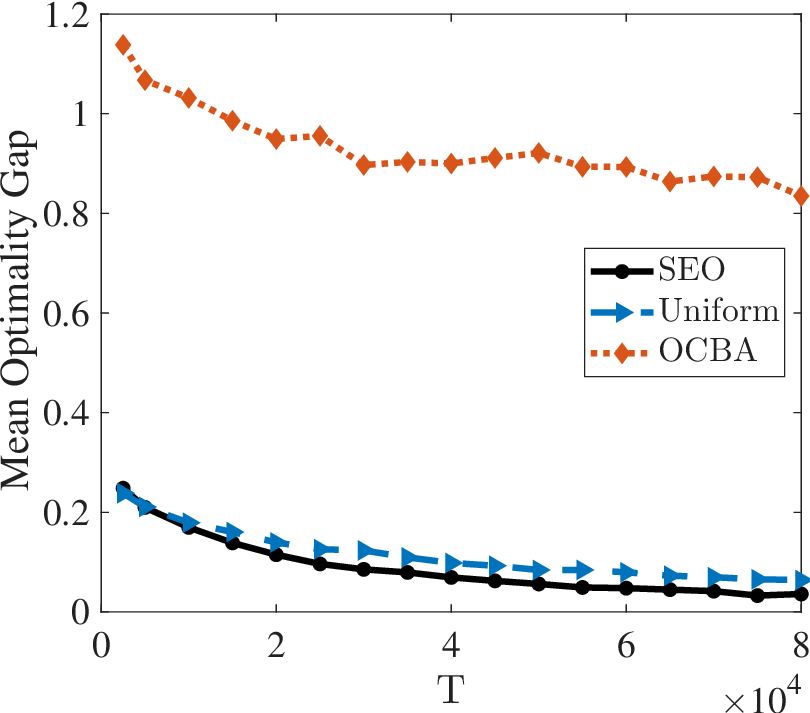}}
\subfigure[$K=40$]{
\label{gK=40:dose} \includegraphics[width=1.91in]{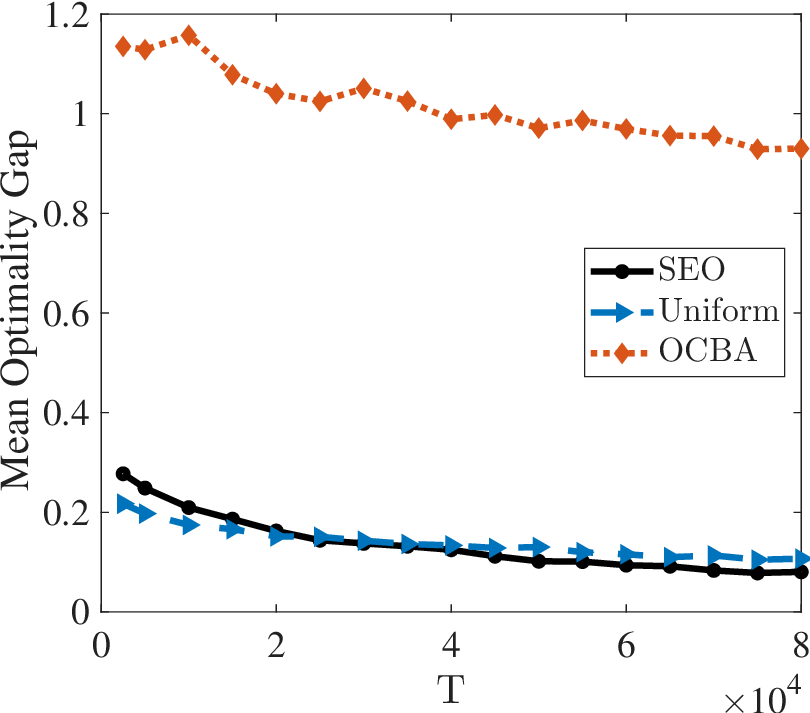}}
\subfigure[$K=128$]{
\label{gK=128:dose} \includegraphics[width=1.91in]{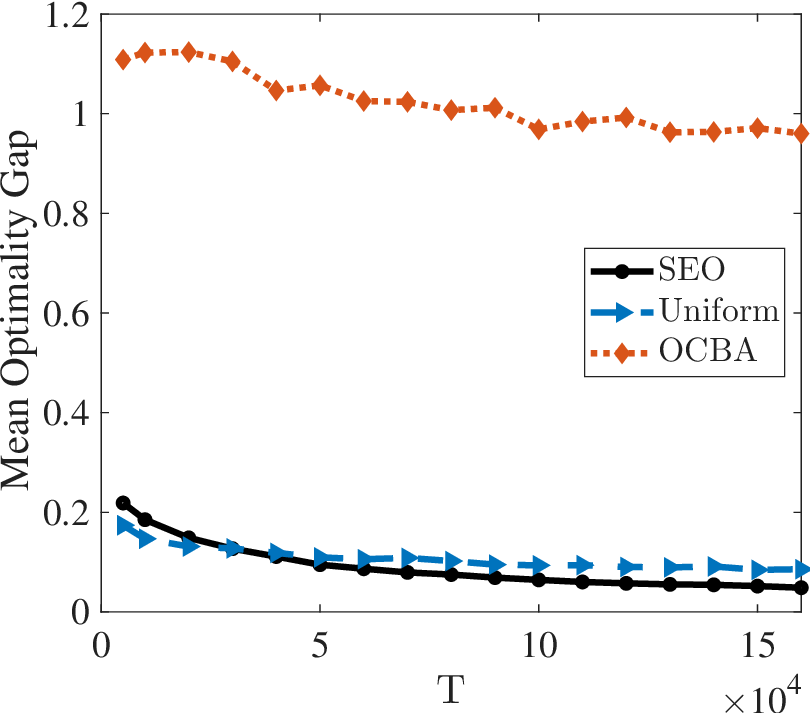}}
\caption{The comparison of optimality gaps between SEO, uniform sampling and OCBA in the optimal dosage problem.}
\label{fig:dosage_gap_comp}
\end{figure}

Since in this setting, we have an analytical optimal solution, we show the mean optimality gap as a function of $T$ in Figure \ref{fig:dosage_gap_comp}, where the mean optimality gap is defined as 
$$\text{Mean Optimality Gap}=\frac{1}{N} \sum_{j=1}^N \bigg[ (a_j x_j^2 + b_j x_j +c_j) - \frac{4a_*c_*-b_*^2}{4a_*} (1+\max_i u_i)\bigg],$$
where $[a_j,b_j,c_j]$ is the selection of the system of the algorithm in experiment $j$, $N$ is the total number of replications, and $x_j$ is the estimated optimal point.
{From Figure \ref{fig:dosage_gap_comp}, overall, SEO and Uniform sampling out perform the OCBA algorithm. 
The advantage of SEO over Uniform becomes more pronounced as the budget increases, suggesting its superior efficiency in directing simulation effort. In contrast, OCBA performs worse, with mean optimality gaps remaining substantially higher and showing only modest improvement with increasing budget.
Furthermore, in very limited budget cases, the SEO algorithm does not outperform uniform sampling. This is because uniform sampling ensures that every system is evaluated equally, whereas the inner gradient descent step in SEO may fail to converge and still exhibit substantial bias when the budget is limited.
}

\subsection{Newsvendor Problem in the selection of the best product}

Newsvendor problems \citet{petruzzi1999pricing,arrow1951optimal,chen2016novel} prevail for decades in revenue management, operations research, and management science, as they are tractable yet still can capture many important realistic characteristics in practice. In the big-data era, date-driven newsvendor problems \citep{huber2019data,ban2019big} gains even more attentions since we can utilize more data to get better estimation of uncertainties, therefore informing better business decisions. However, data collection (or data purchasing) can be costly in practice, yielding a need to intelligently collecting data to achieve high-quality decisions.

In this example, we assume there are $K$ products as contenders. Each product is a system, having their own price, cost structure and different demand function. Specifically, we consider the function classes
$\mathcal{F}_{p,c}=\left\{ f_{q}|f_{q}(x)=cq-p\min \left\{
q,x\right\},q\geq0 \right\},$
where $p$ is the price and $c$ denotes the cost. We consider the newsvendar problem $
v_i = \sup_{g \in \mathcal{F}_{p_i,c_i}} \mathbb{E}_{P^i}[g(X)]=\sup_{q\geq0}\mathbb{E}_{P^i}[p\min \left\{
q,x\right\} - cq],$
where $p_i,c_i,P_i$ stands for the price, cost and the demand distribution for the $i$-th product. Here, we assume $p_i$ and $c_i$ are known beforehand. However, $P_i$ is unknown and we have access to collect samples from the distribution $P_i$. The goal is to find the best product that gives us the most profits. 

For the experiment specifics, we assume that $p_i = \tfrac{1}{2}i + 5$ and $c_i= \tfrac{1}{5} i + 1$ and the $P_i$ follows $\mathrm{Poi}(\lambda_i)$, a Poisson distribution with rate $\lambda_i = 250 - 6\times i$. The problem (\ref{prob:saa}) can be solved in closed form since the empirical optimal solution is the $(p_i-c_i)/p_i$ quantile of the empirical distribution. We also report the mean optimality gaps of SEO and Uniform. The optimality gap is defined as 
$$\text{Mean Optimality Gap} = v^* - \frac{1}{N} \sum_{j=1}^N \mathbb E_{X \sim P^{k_j}}[p_{k_j} \min \{q_j,X\} - c_{k_j} q_j],$$
where $k_j$ is the selection of the algorithm in the $j$th experiment and $v^*$ is the overall optimal value.

Figures \ref{fig:newsvendor} shows the comparison between SEO and uniform sampling algorithms in the Newsvendor problem. It is clear that our algorithm has consistently higher probability of correct selection than one of the uniform sampling algorithm. Further, superiority is even more significant when the number of product $K$ is large. We do not compare our algorithm with OCBA since the focus here is the sample-efficiency. Once we collect samples that reflect the unknown underlying distributions, the optimization part is easy and straightforward. Therefore, it is not relevant to discretize the decision space inside each system and perform OCBA.  \ref{ec:numerical:newsvendor} contain additional plots. From Figure \ref{fig:newsvendor1} and  \ref{fig:newsvendor}, it is shown that SEO outperforms uniform sampling for all  values of system number $K$ and budget $T$.
\begin{figure}[htbh]
\centering
\subfigure[$K=16$]{
\label{K=16:newsvendor} \includegraphics[width=1.91in]{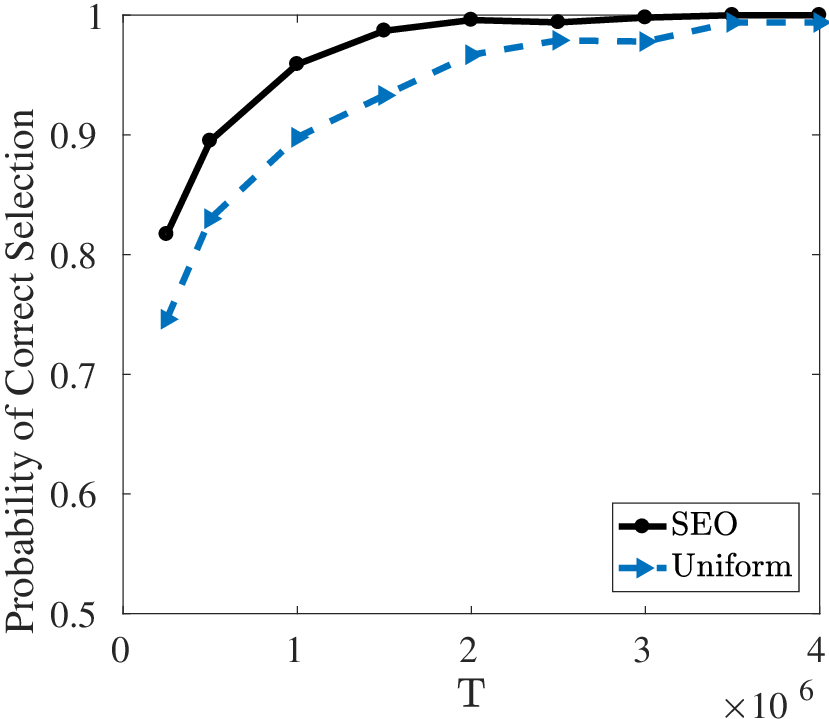}}
\subfigure[$K=40$]{
\label{K=40:newsvendor} \includegraphics[width=1.91in]{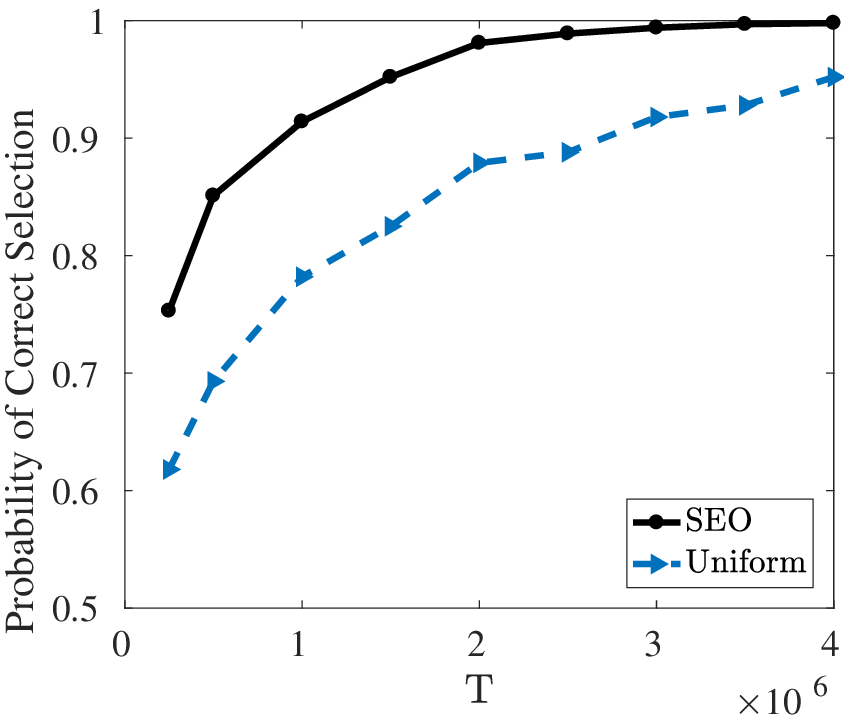}}
\subfigure[$K=128$]{
\label{K=128:newsvendor} \includegraphics[width=1.91in]{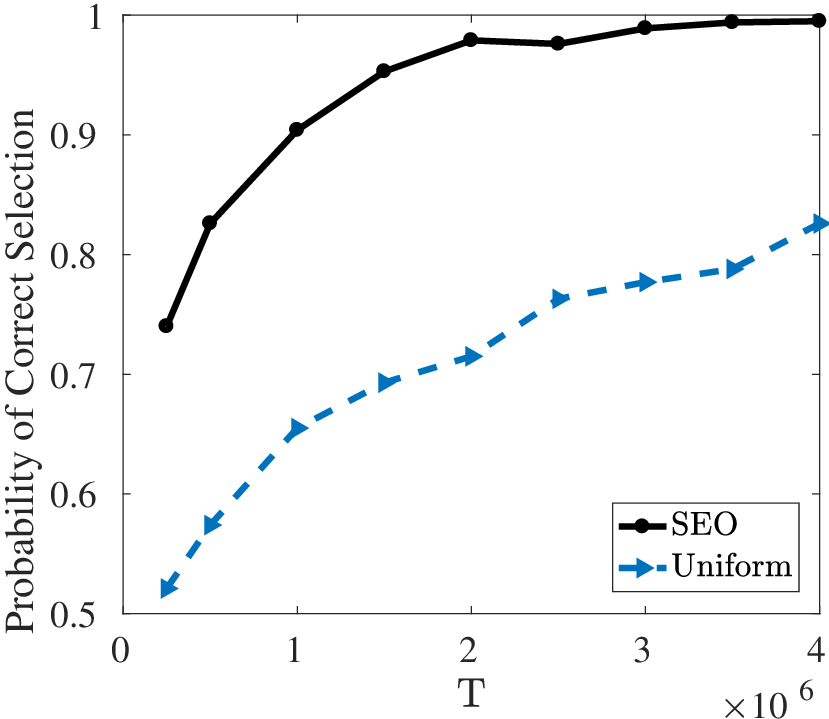}}
\caption{The comparison between SEO and uniform sampling in the newsvendor problem.}
\label{fig:newsvendor1}
\end{figure}

\begin{figure}[htbh]
\centering
\subfigure[$K=16$]{
\label{gK=16:newsvendor} \includegraphics[width=1.91in]{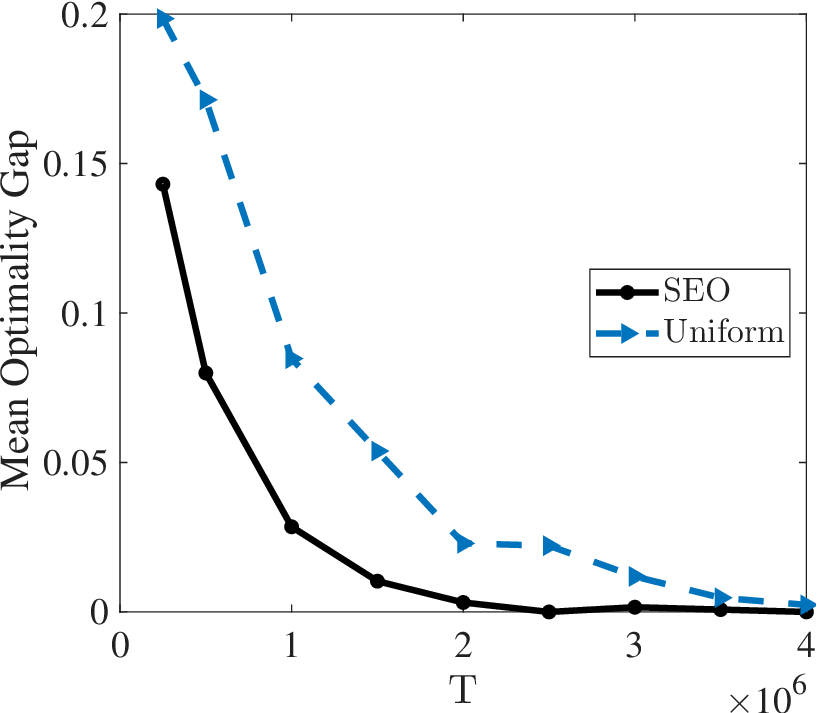}}
\subfigure[$K=40$]{
\label{gK=40:newsvendor} \includegraphics[width=1.91in]{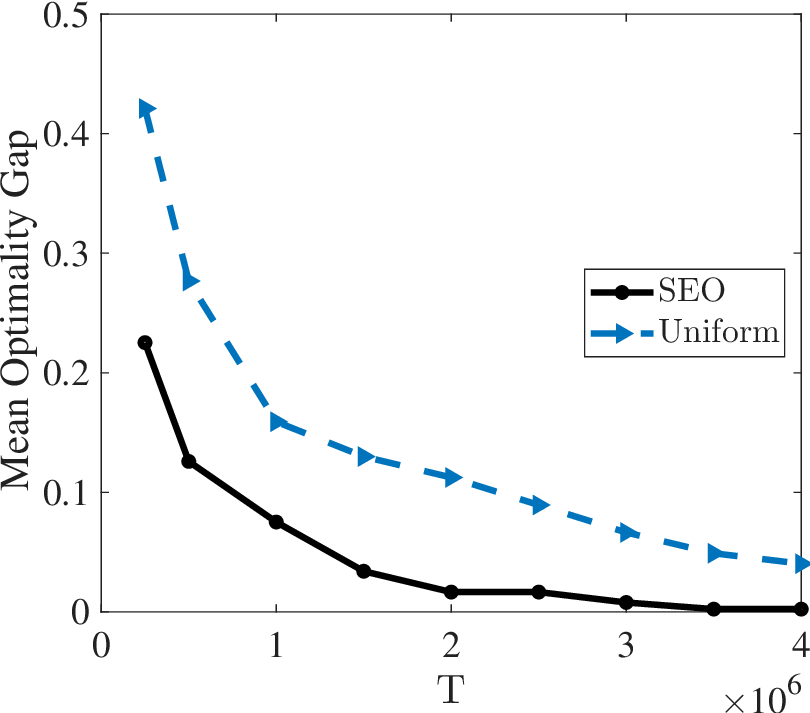}}
\subfigure[$K=128$]{
\label{gK=128:newsvendor} \includegraphics[width=1.91in]{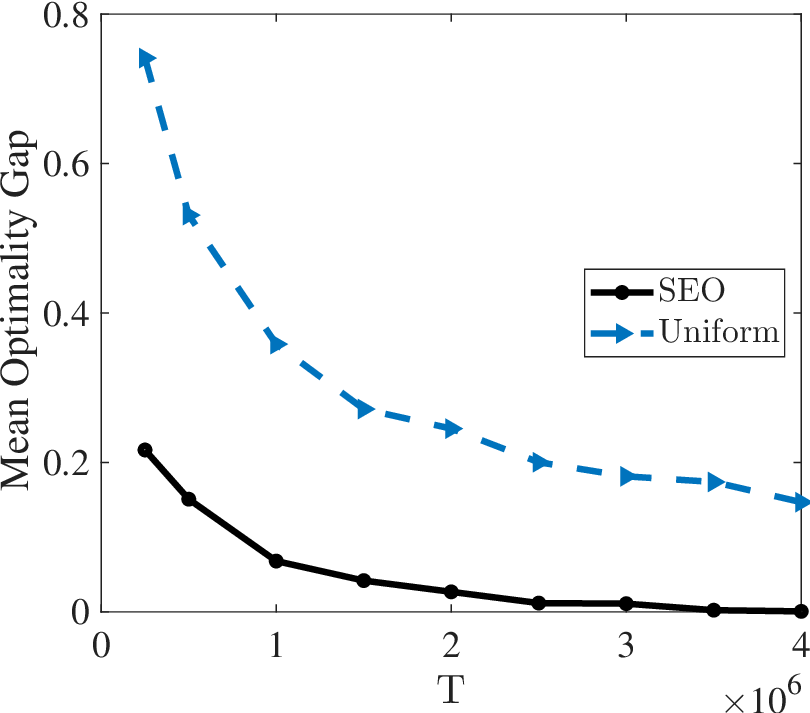}}
\caption{The comparison between the optimality gaps of SEO and uniform sampling in the newsvendor problem.}
\label{fig:newsvendor}
\end{figure}

\section{Conclusion}
In this work, we formulate and solve a class of problems termed selecting the best optimizing system (SBOS). We propose a simple algorithm that sequentially allocates samples across systems and identifies the best one under a fixed sampling budget. The algorithm combines sequential elimination with stochastic gradient descent, and we show that its probability of false selection decays exponentially with the budget. 
Future work includes two directions: (i) extending SBOS to the fixed-precision (fixed-confidence) framework, which poses different challenges for design and analysis, and (ii) developing distributionally robust formulations to address system non-stationarities, where models today may differ from those in the near future.

\bibliographystyle{plainnat} 
\bibliography{reference}

\setcounter{section}{0} \setcounter{subsection}{0} \setcounter{equation}{0}
\ \renewcommand
{\thesection}{Appendix \Alph{section}}
\renewcommand\thesubsection{\thesection.\arabic{subsection}}
\renewcommand{\theequation}{\Alph{section}.\arabic{equation}}
\renewcommand{\thelemma}{\Alph{section}\arabic{lemma}}
\renewcommand{\thetheorem}{\Alph{section}\arabic{theorem}}
\section{Optimal Computing Budget Allocation (OCBA) Algorithm}
We give the details of the Optimal Computing Budget Allocation (OCBA) Algorithm in this section.
\label{EC.1:algo}
\begin{breakablealgorithm}
	\caption{Optimal Computing Budget Allocation (OCBA) for Optimizing Systems}
	\label{alg:OCBA}
	\begin{algorithmic}[1]
		\State \textbf{Input: } $N$, $K$, the discretization $d$-element $\mathcal{X}^d = [x_1,\ldots,x_d]$ of the feasible set (can be a function of $i$),  and the initial phase proportion $\alpha_0$.
		\State \textbf{Initalization:} Set $N_0=\max\{2, \alpha_0 T/K/d\}$ and let $\ell = N_0  K  d$. Sample each system with each possible decision $x_j$ in the set $\mathcal{X}^d$, $N_0$ times and compute mean $\bar{X}_{ij}(\ell)$ and variance $S^2_{ij}(\ell)$ for $i \in\{1,2,\ldots,K\}$ and $j \in [d]$. Let $N_{ij}(\ell) =N_0$ for $i \in\{1,2,\ldots,K\}$ and $j \in [d]$.
		\While{$\ell < T $}
		
		\State Compute the best system with the best decision $\{\hat{b}^s,\hat{b}^d\} \gets \arg\max_{i\in\{1,2,\ldots,K\},j\in[d]} \bar{X}_{ij}$.
		\State Set $s\gets0$.
		\For{$\{i,j\} =\{1,1\}\ldots \{K,d\}$}
		\If{$\{i,j\} \neq \{\hat{b}^s,\hat{b}^d\}$}
		
		\State Compute $\hat{\beta}_{ij}\gets S^2_{ij}(\ell) / [\bar{X}_{\hat{b}^s\hat{b}^d}-\bar{X}_{ij}]^2$.
		\State Compute $s\gets s + \hat{\beta}_{ij}^2/S^2_{ij}(\ell) $;
		\EndIf
			\EndFor
		\State Compute $\hat{\beta}_{\hat{b}^s\hat{b}^d} \gets S_{\hat{b}^s\hat{b}^d}(\ell) \sqrt{s}.$
		\State Run one replication for the system with the decision $\{i^*,j^*\}\gets \arg\max_{i\in\{1,2,\ldots,K\},j \in[d]}\hat{\beta}_{ij}/N_{ij}(\ell)$.
		\State $N_{i^*j^*}(\ell + 1) \gets N_{i^*j^*}(\ell ) + 1$, and update $\bar{X}_{i^*j^*}(\ell+1)$ and $S^2_{i^*j^*}(\ell+1)$.
		\State $\ell \gets \ell +1$.
		\EndWhile
		\State \textbf{Output: } Compute $\{\hat{b}^s,\hat{b}^d\} \gets \arg\max_{i\in\{1,2,\ldots,K\},j\in[d]} \bar{X}_{ij}(T)$ and return $\hat{b}^s$.
	\end{algorithmic}
\end{breakablealgorithm}

\medskip

To ensure the probability of false selection converging to zero when $T\rightarrow +\infty$, we show that technically the cardinality $d$ of the discretization set needs to be very large in Lemma \ref{lma:ocba_d}. Lemma \ref{lma:ocba_d} shows that the cardinality $d$ of the discretization set $\mathcal X^d$ cannot be small for large $T$, or Algorithm 3 is not asymptotically optimal.
\begin{lemma}
\label{lma:ocba_d}
We consider a simulation optimization regime. Let $\mathcal{O}_{M,\mathcal{X},K,a}$ contain all SBOS instances with K systems, Lipschitz constant less than $M$, the complexity term $H$ less than $a$. and the inner-layer decision space $\mathcal{X}$. Then, if $d<\frac{1}{2}{MD_{\mathcal{X}}\sqrt{a}}$ -1, there exists an instance in $\mathcal{O}_{M,\mathcal{X},K,a}$ such that the probability of false selection of Algorithm \ref{alg:OCBA} is at least $1/2$.
\end{lemma}
The proof of Lemma \ref{lma:ocba_d} is in \ref{proof:lma:ocba}.
\section{Proofs of Statements}
\subsection{Proof of results in Section \ref{sec:theory:sim}} \label{ec:proof_sec_2}
\begin{proof}{Proof of Proposition \ref{prop:conv}}
Note that
$$
	\left\vert \hat{v}_{i}^{T}-v_{i}\right\vert =\left\vert \hat{v}_{i}^{T}-%
	\frac{1}{T}\sum_{t=1}^{T}f_{i}(x_{t})+\frac{1}{T}%
	\sum_{t=1}^{T}f_{i}(x_{t})-v_{i}\right\vert
	\leq \underbrace{\left\vert \hat{v}_{i}^{T}-\frac{1}{T}%
		\sum_{t=1}^{T}f_{i}(x_{t})\right\vert }_{(a)}+\underbrace{%
		v_{i}-\frac{1}{T}\sum_{t=1}^{T}f_{i}(x_{t})}_{(b)}.
$$
For part(a), we define a martingale $\{Z_{n}\}_{n=1}^{T}$ with $Z_{0}=0$ and
\begin{equation*}
	Z_{n}=Z_{n-1}+F_{i}(x_{t},\xi _{t})-\mathbb{E}[F_{i}(x_{t},\xi _{t})|\xi
	_{1:i-1}]=Z_{n-1}+F_{i}(x_{t},\xi _{t})-f_{i}(x_{t}).
\end{equation*}%
We have
\begin{equation*}
	\exp (\lambda \left( Z_{n}-Z_{n-1}\right) |\xi _{1:i-1})\leq \exp \left(
	\frac{\lambda ^{2}\sigma _{F,i}^{2}}{2}\right) .
\end{equation*}%
Then, by Azuma's inequality, we have, for any $\epsilon _1>0$, 
\begin{equation*}
	\mathbb{P}\left( \underbrace{\left\vert \hat{v}_{i}^{T}-\frac{1}{T}%
		\sum_{t=1}^{T}f_{i}(x_{t})\right\vert } \geq \epsilon_1 \right)=\mathbb{P}%
	\left( \left\vert \frac{Z_{T}}{T}\right\vert \geq \epsilon _{1}\right) \leq
	2\exp \left( -\frac{2T\epsilon _{1}^{2}}{\sigma _{F,i}^{2}}\right) .
\end{equation*}%

{For part (b), recall that $y_{t+1} = x_t+\gamma G(x_t,\xi_t)$, $x_{t+1} = \Pi_{\mathcal X}(y_{t+1})$. Then, by the contraction property of projection map, we have}
$$\begin{aligned}
    \|x_{t+1}-x^*\|^2 & = \| \Pi_{\mathcal X}(y_{t+1}) - \Pi_{\mathcal X}(x^*)\|^2 \\
    & \leq \|y_{t+1} - x^* \|^2 \\
    & = \| x_t + \gamma G(x_t,\xi_t) - x^*\|^2  \\
    & = \|x_t-x^* \|^2 + \gamma^2 \|G(x_t,\xi_t)\|^2 + 2\gamma \left \langle G(x_t,\xi_t),x_t-x^*\right\rangle.
\end{aligned}$$
Adding the above up from $t=1$ to $T$, we have
$$\| x_{T+1} - x^*\|^2 - \| x_1 - x^*\|^2 \leq \gamma^2 \sum_{t=1}^T \|G(x_t,\xi_t) \|^2 + 2\gamma \sum_{t=1}^T \left\langle G(x_t,\xi_t),x_t-x^*\right\rangle.$$
Then, note that (4.1.10) in \cite*{lan2020first} is also vaild for $\frac{1}{T}%
\sum_{t=1}^{T}f_{i}(x_{t})$ in the sense that
$$
\begin{aligned}
	&v_{i}-\frac{1}{T}\sum_{t=1}^{T}f_{i}(x_{t})\leq \frac{1}{T}\left[ \frac{1}{%
		2}\left( x_{1}-x^{\ast }\right) ^{2}-\frac{1}{2}\left( x_{T+1}-x^{\ast
	}\right) ^{2}\right. \\
	&+\left. 2\sum_{t=1}^{T}\gamma ^{2}(M^{2}+\left\Vert G_{i}(x_{t},\xi
	_{t})-f_{i}^{\prime }(x_{t})\right\Vert _{2}^{2}-\sum_{t=1}^{T}\gamma
	\left\langle G_{i}(x_{t},\xi _{t})-f_{i}^{\prime
	}(x_{t}),x_{t}-x^*\right\rangle \right] .
    \end{aligned}
$$
And Let $\delta _{t}=G_{i}(x_{t},\xi _{t})-f_{i}^{\prime }(x_{t}).$\ Since
the step-size is constant, after applying Markov inequality and Assumption \ref{assump:gd}.3, the (4.1.18) in \cite*{lan2020first} can be rewritten as
\begin{align}
	\mathbb{P}\left\{ \sum_{t=1}^{T}\left\Vert \delta _{t}\right\Vert
	_{2}^{2}/\sigma _{G,i}^{2}>(1+\lambda )T\right\} &\leq \exp (-(1+\lambda )T)%
	\mathbb{E}\left[ \exp \left( \sum_{t=1}^{T}\left\Vert \delta _{t}\right\Vert
	_{2}^{2}/\sigma _{G,i}^{2}\right) \right]  \label{4.1.18_reformulate2} \\
	&\leq \exp \left( -\lambda T\right) .  \notag
\end{align}

Then, by similar lines with \citet[Proposition 4.1]{lan2020first}, we have
\begin{equation*}
	\mathbb{P}\left( v_{i}-\frac{1}{T}\sum_{t=1}^{T}f_{i}(x_{t})\geq \frac{3D_{\mathcal{X%
		}}\sqrt{M^{2}+\sigma ^{2}}}{\sqrt{T}}+\epsilon _{2}\right) \leq \exp
	(-\epsilon _{2}T\sqrt{T}/\left( 3\sigma _{G,i}D_{\mathcal{X}_i}\right) )+\exp
	\left( -\epsilon _{2}^{2}T/\left( 27\sigma _{G,i}^{2}D_{\mathcal{X}%
	}^{2}\right) \right) .
\end{equation*}%
We let $\epsilon _{1}=\frac{\left( \sigma _{F,i}/\sqrt{2}\right) \epsilon }{3%
	\sqrt{3}\sigma _{G,i}D_{\mathcal{X}_i}+\sigma _{F,i}/\sqrt{2}}$ and $\epsilon
_{2}=\frac{3\sqrt{3}\sigma _{G,i}D_{\mathcal{X}_i}\epsilon }{3\sqrt{3}\sigma
	_{G,i}D_{\mathcal{X}_i}+\sigma _{F,i}/\sqrt{2}}.$ Then, we have
\begin{equation*}
	\frac{\epsilon _{2}^{2}}{27\sigma _{G,i}^{2}D_{\mathcal{X}_i}^{2}}=\frac{%
		2\epsilon _{1}^{2}}{\sigma _{F,i}^{2}}=\frac{\epsilon ^{2}}{\left( 3\sqrt{3}%
		\sigma _{G,i}D_{\mathcal{X}_i}+\sigma _{F,i}/\sqrt{2}\right) ^{2}}.
\end{equation*}
\end{proof}
\begin{proof}{Proof of Theorem \ref{thm:sim}.}
	
			We rely on the proof of Theorem 33.10 in \cite*{lattimore2020bandit}. It is easy to note that $|%
			\mathcal{A}_{\ell }|=\lfloor K2^{1-\ell } \rfloor$ and $T_{\ell }\geq T2^{\ell -1}/\lfloor\log
			_{2}\left( K\right) \rfloor/K.$ We first observe that
$$
				\mathbb{P}(\hat{v}_{1}^{T_\ell } \leq \hat{v}_{i}^{T_\ell }|\left\{
				i,1\right\} \subset \mathcal{A}_{\ell })
				\leq \mathbb{P}\left( \hat{v}_{1}^{T_\ell }\leq v_{1}+\Delta _{i}/2\right) +%
				\mathbb{P}\left( \hat{v}_{i}^{T_\ell }\geq v_{i}-\Delta _{i}/2\right) .
$$
		Since $
			T\geq \lfloor\log _{2}\left( K\right)\rfloor K\left( \max_{i\in \lbrack K]}\frac{24D_{\mathcal{X}_i}}{\Delta
				_{i}}\left( \sqrt{M^{2}+\sigma _{G,i}^{2}}\right) \right)
			^{2}
		$, we have
		\[
		\frac{3D_{\mathcal{X}_i}\sqrt{M^{2}+\sigma _{G,i}^{2}}}{\sqrt{T_l}} \leq \Delta_i/4.
		\]
		 By letting $%
			\epsilon =\Delta _{i}/4,$ and applying Proposition \ref{prop:conv} we have
			\begin{align}
				\mathbb{P}&\left( \left\vert \hat{v}_{i}^{T_\ell}-v_{i}\right\vert \geq \Delta_i/4 \right)				
				\\ &\leq 3\exp \left( -\frac{T_\ell\epsilon ^{2}}{16\left( 3\sqrt{3}\sigma _{G,i}D_{%
						\mathcal{X}_i}+\sigma _{F,i}/\sqrt{2}\right) ^{2}}\right) +\exp \left( -\frac{%
					\Delta_i T_\ell\sqrt{T_\ell}}{4(3\sigma _{G,i}D_{\mathcal{X}_i}+\sigma _{F,i}/\sqrt{6})}%
				\right) ,
			\end{align}
		which in turn gives us
		\begin{equation}
			\label{eq:single_wrong}
				\mathbb{P}(\hat{v}_{1}^{T_\ell } \leq \hat{v}_{i}^{T_\ell }|\left\{
				i,1\right\} \subset \mathcal{A}_{\ell })  \leq 6\exp \left( -\frac{T_\ell\Delta_i ^{2}}{48M_\sigma ^{2}}\right) +2\exp \left( -\frac{%
					\Delta_i T_\ell\sqrt{T_\ell}}{4M_\sigma} \right).%
			\end{equation}
			Next, we define a new set
			\begin{equation*}
				\mathcal{A}_{\ell }^{\prime }=\left\{ i\in \mathcal{A}_{\ell }\left\vert
				\sum_{j\leq i}\mathbb{I}\left\{ j\in \mathcal{A}_{\ell }\right\} >\lfloor \left\vert \mathcal{A}%
				_{\ell }\right\vert /4\rfloor \right. \right\} ,
			\end{equation*}%
			to be the bottom (ordered by true value) three-quarters of the systems in round $\ell .$ Then, if the
			optimal system is eliminated in this round, we must have
			\begin{equation}
				N_{\ell }=\sum_{i\in \mathcal{A}_{\ell }^{\prime }}\mathbb{I}\{\hat{v}%
				_{1}^{T_\ell }\geq \hat{v}_{i}^{T_\ell }\}\geq \left\lfloor\frac{1}{3}(|\mathcal{A}_{\ell
				}^{\prime }|+1)\right\rfloor.  \label{eqn:N_l2}
			\end{equation}%
			On the other hand, by applying the union bound to the bound (\ref{eq:single_wrong}),
			\begin{equation*}
				\mathbb{E}\left[ N_{\ell }\right] \leq |\mathcal{A}_{\ell }^{\prime
				}|\max_{i\in \mathcal{A}_{\ell }^{\prime }} \left\lbrace 6\exp \left( -\frac{T_\ell\Delta_i ^{2}}{48M_\sigma ^{2}}\right) +2\exp \left( -\frac{%
				\Delta_i T_\ell\sqrt{T_\ell}}{4M_\sigma} \right)\right\rbrace  .
			\end{equation*}%
			Let $i_{\ell }^{\prime }\triangleq \min \mathcal{A}_{\ell }^{\prime }$ $\geq\lfloor\left\vert
			\mathcal{A}_{\ell }\right\vert /4\rfloor +1 \geq K2^{-1-\ell }.$
			 Then, we have%
			\begin{align}
				\mathbb{E}\left[ N_{\ell }\right] &\leq |\mathcal{A}_{\ell }^{\prime
				}|\max_{i\in \mathcal{A}_{\ell }^{\prime }} \left\lbrace 6\exp \left( -\frac{ 2^{\ell+1}T\Delta_i ^{2}}{192K \lfloor\log_2(K)\rfloor M_\sigma ^{2}}\right) +2\exp \left( -\frac{%
					\Delta_i 2^{\ell+1}T\sqrt{T2^{\ell-1}/\lfloor\log_2(K)\rfloor / K}}{16 K \lfloor \log_2(K) \rfloor M_\sigma} \right)\right\rbrace   \\
				 &\leq|\mathcal{A}_{\ell }^{\prime
				 }| \max_{i\in \mathcal{A}_{\ell }^{\prime }} \left\lbrace 6\exp \left( -\frac{ T\Delta_{i_{\ell }^{\prime }} ^{2}}{192 \lfloor\log_2(K)\rfloor M_\sigma ^{2}i_{\ell }^{\prime }}\right) +2\exp \left( -\frac{%
				 	 T\sqrt{T}\Delta_{i_{\ell }^{\prime }}}{16  \sqrt{K \lfloor \log_2(K) \rfloor^2} M_\sigma i_{\ell }^{\prime }} \right)\right\rbrace .
				\label{ineq:N_L}
			\end{align}%
			By combining (\ref{eqn:N_l2}) and (\ref{ineq:N_L}), we have
			\begin{align*}
				\mathbb{P}(1 \notin \mathcal{A}_{\ell +1}|1\in \mathcal{A}_{\ell })&\leq
				\mathbb{P}\left( N_{\ell }\geq  \left\lfloor\frac{1}{3}(|\mathcal{A}_{\ell
				}^{\prime }|+1)\right\rfloor\right)
				\leq \frac{\mathbb{E}[N_{\ell }]}{ \left\lfloor(|\mathcal{A}_{\ell
					}^{\prime }|+1)/3\right\rfloor} \\
				&\leq4\max_{i\in \mathcal{A}_{\ell }^{\prime }} \left\lbrace 6\exp \left( -\frac{ T\Delta_{i_{\ell }^{\prime }} ^{2}}{192 \lfloor\log_2(K)\rfloor M_\sigma ^{2}i_{\ell }^{\prime }}\right) +2\exp \left( -\frac{%
					T\sqrt{T}\Delta_{i_{\ell }^{\prime }}}{16  \sqrt{K \lfloor \log_2(K) \rfloor^2} M_\sigma i_{\ell }^{\prime }} \right)\right\rbrace .
			\end{align*}%
			Finally, by adding the above inequality for all $\ell $ from $1$ to L, we
			have		
			\begin{align}
				\mathbb{P}(1 \notin &\mathcal{A}_{L+1})\leq \sum_{\ell =1}^{L}\mathbb{P}%
				(1\notin \mathcal{A}_{\ell +1}|1\in \mathcal{A}_{\ell }) \\
				&\leq \lfloor \log_2(K) \rfloor \left\lbrace 24\exp \left( -\frac{ T }{192 \lfloor\log_2(K)\rfloor M_\sigma ^{2}H_{2}(v)}\right) +8\exp \left( -\frac{%
					T\sqrt{T}}{16  \sqrt{K \lfloor \log_2(K) \rfloor^2} M_\sigma H'_{2}(v)} \right)\right\rbrace ,
			\end{align}%
			where $H_{2}(v)=\max_{i>1}\frac{i}{\Delta _{i}^{2}},$ and $H'_{2}(v)=\max_{i>1}\frac{i}{\Delta _{i}}$.
			
			$ \hfill \square$
		\end{proof}

\begin{proof}{Proof of Proposition \ref{thm:ld:queue}.}
We rely on the proof ideas from \cite*{carpentier2016tight}. We construct $K$ hard instances. Let $(p_k)_{2 \leq k \leq K}$ be $(K-1)$ real numbers in $[1/4,1/2)$. Let $p_1=1/2$. We denote $\mathrm{Ber}(p)$ to the Bernoulli distribution with probability $p$. Then, for the $k$-th system of the $i$-th instance, we assume $F_i(x,\xi)$ follows distribution $2 \bar{\sigma} \mathrm{Ber}(p_k)$ if $i\neq k$. Otherwise, if $i=k$, we assume $F_i(x,\xi)$ follows distribution $2 \bar{\sigma} \mathrm{Ber}(1-p_k)$. Then, all systems for all instances have the variance less or equal than $\mathcal{O}$. Following the notions in \cite*{carpentier2016tight}, let $d_k = 1/2-p_k$. Then, for the $i$-instance, the complexity is  $H(i)=1/4\bar{\sigma}^{-2}\sum _{1\leq k \leq K, k\neq i}(d_i+d_k)^{-2}$. This problem is essentially the same as the problem in \cite*{carpentier2016tight}. Then, Theorem 1 in \cite*{carpentier2016tight} gives the desired result. $\hfill \square$
\end{proof}
\subsection{Proof of results in Section \ref{sec:theory:SA}}
\label{ec:proof:SA}
We first collect some useful results.
\begin{definition}[Rademacher complexity]
	Let $\mathcal{F}$ be a family of real-valued functions $f:Z\rightarrow
	\mathbb{R}$ Then, the Rademacher complexity of $\mathcal{F}$ is defined as
	\begin{equation*}
		\mathcal{R}_{n}\left( \mathcal{F}\right) \triangleq \mathbb{E}_{z,\sigma }%
		\left[ \sup_{f\in \mathcal{F}}\left\vert \frac{1}{n}\sum_{i=1}^{n}\sigma
		_{i}f(z_{i})\right\vert \right] ,
	\end{equation*}%
	where $\sigma _{1},\sigma _{2},\ldots ,\sigma _{n}$ are i.i.d with the
	distribution $\mathbb{P}\left( \sigma _{i}=1\right) =\mathbb{P}\left( \sigma
	_{i}=-1\right) =1/2.$
	\label{def:rad}
\end{definition}

\begin{theorem}[Theorem 4.10 in \protect\cite*{wainwright2019high}]
	\label{thm:rademacher}
	
	If $f(z)\in [-B,B],$ we have with probability at least $1-\exp \left(- \frac{%
		n\epsilon^2}{2B^2} \right) $,
	\begin{equation*}
		\sup_{f\in \mathcal{F}}\left| \frac{1}{n}\sum_{i=1}^{n}f(z_{i})-\mathbb{E}
		f(z)\right| \leq 2\mathcal{R}_{n}\left( \mathcal{F}\right) +\epsilon.
	\end{equation*}
\end{theorem}

\begin{theorem}[Dudley's Theorem, (5.48) in \protect\cite*{wainwright2019high}%
	]
	If $f(z)\in \lbrack -B,B],$ we have a bound for the Rademacher complexity, %
	\label{Dudley}
	\begin{equation*}
		\mathcal{R}_{n}\left( \mathcal{F}\right) \leq  \frac{24\mathcal{J}(\mathcal{F},P) }{%
			\sqrt{n}} ,
	\end{equation*}%
	where $N(t,\mathcal{F}$,$\left\Vert \cdot \right\Vert _{P _{n}})$ is $t$%
	-covering number of set $\mathcal{F}$ and
\end{theorem}
\begin{proof}{Proof of Theorem \ref{thm:SA}}
	Let $g_{\ast }^{i}=\arg \max_{g\in \mathcal{F}_{i}}\mathbb{E}_{P^{i}}\left[
	g(X)\right] .$ Then,%
	\begin{align}
		v_{i}-\hat{v}_{i}^{n_i}&=\mathbb{E}_{P^{i}}%
		\left[ g_*^{i}(X)\right]  -\mathbb{E}_{P_{n_{i}}^{i}}\left[ g_{\ast }^{i}(X)\right] +\mathbb{E}_{P_{n_{i}}^{i}}\left[ g_{\ast }^{i}(X)\right] -\sup_{g\in \mathcal{F}_{i}}\mathbb{E}_{P_{n_{i}}^{i}}%
		\left[ g(X)\right]
		\\
		&\leq \sup_{g\in \mathcal{F}_{i}}\left\vert \mathbb{E}_{P_{n_{i}}^{i}}\left[
		g(X)\right] -\mathbb{E}_{P^{i}}\left[ g(X)\right] \right\vert .
	\end{align}%
	Similarly, we have
	\begin{equation*}
		\hat{v}_{i}^{n_i}-v_{i}\leq \sup_{g\in \mathcal{F}_{i}}\left\vert \mathbb{E}%
		_{P_{n_{i}}^{i}}\left[ g(X)\right] -\mathbb{E}_{P^{i}}\left[ g(X)\right]
		\right\vert .
	\end{equation*}
	
	We define $\mathcal{R}_{n}$ as the Rademacher complexity (Definition \ref{def:rad}). By Theorems \ref{thm:rademacher} and \ref{Dudley}, we have with probability at least $1-\exp
	\left( -\frac{n_i\epsilon ^{2}}{2B^{2}}\right) $,
	\begin{equation*}
		\left\vert v_{i}-\hat{v}_{i}^{n_i}\right\vert \leq \sup_{f\in \mathcal{F}%
			_{i}}\left\vert \frac{1}{n_i}\sum_{i=1}^{n_i}g(z_{i})-\mathbb{E}g(z)\right\vert
		\leq 2\mathcal{R}_{n_i}\left( \mathcal{F}_{i}\right) +\epsilon \leq \frac{48\mathcal{J}(\mathcal{F},P) }{%
			\sqrt{n_i}} +\epsilon,
	\end{equation*}%
	which means
	\begin{equation*}
		\mathbb{P}\left( \hat{v}_{i}^{n_i}\notin \left( v_{i}-\frac{48\mathcal{J}(\mathcal{F},P) }{%
			\sqrt{n_i}}-\epsilon ,v_{i}+\frac{48\mathcal{J}(\mathcal{F},P) }{%
			\sqrt{n_i}} +\epsilon \right) \right) \leq \exp \left( -\frac{n_i\epsilon ^{2}%
		}{2B^{2}}\right) .
	\end{equation*}
	Note that $|%
	\mathcal{A}_{\ell }|=\lfloor K2^{1-\ell } \rfloor$ and $T_{\ell }\geq T2^{\ell -1}/\lfloor\log
	_{2}\left( K\right) \rfloor/K.$ Similar with the proof of Theorem \ref{thm:sim}, we observe
	$$
		\mathbb{P}(\hat{v}_{1}^{T_\ell } \leq \hat{v}_{i}^{T_\ell }|\left\{
		i,1\right\} \subset \mathcal{A}_{\ell })
		\leq \mathbb{P}\left( \hat{v}_{1}^{T_\ell }\leq v_{1}+\Delta _{i}/2\right) +%
		\mathbb{P}\left( \hat{v}_{i}^{T_\ell }\geq v_{i}-\Delta _{i}/2\right) .
$$
	Since $	T\geq \lfloor\log _{2}\left( K\right)\rfloor K\left(  \max_{i\in [ K]}\frac{192}{\Delta
		_{i}}\left(\mathcal{J}(\mathcal{F}_i,P^i)\right) \right)
	^{2},$
	we have
	\[
	\frac{48\mathcal{J}(\mathcal{F},P) }{%
		\sqrt{n_i}} \leq \Delta_i/4.
	\]
	
	 By letting $%
	\epsilon =\Delta _{i}/4,$ we have
	\begin{equation}
		\mathbb{P}(\hat{v}_{1}^{T_\ell }\geq \hat{v}_{i}^{T_\ell }|\left\{ i,1\right\}
		\subset \mathcal{A}_{\ell })\leq 2\exp \left( -\frac{T_{\ell }\Delta _{i}^{2}%
		}{32B^{2}}\right) .
	\label{eq:single_wrong_sa}
	\end{equation}%
		Next, we define a new set
	\begin{equation*}
		\mathcal{A}_{\ell }^{\prime }=\left\{ i\in \mathcal{A}_{\ell }\left\vert
		\sum_{j\leq i}\mathbb{I}\left\{ j\in A\right\} >\lfloor \left\vert \mathcal{A}%
		_{\ell }\right\vert /4\rfloor \right. \right\} ,
	\end{equation*}%
	to be the bottom (ordered by true value) three-quarters of the systems in round $\ell .$ Then, if the
	optimal system is eliminated in this round, we must have
	\begin{equation}
		N_{\ell }=\sum_{i\in \mathcal{A}_{\ell }^{\prime }}\mathbb{I}\{\hat{v}%
		_{1}^{T_\ell }\geq \hat{v}_{i}^{T_\ell }\}\geq \left\lfloor\frac{1}{3}(|\mathcal{A}_{\ell
		}^{\prime }|+1)\right\rfloor.  \label{eqn:N_lSA}
	\end{equation}%
	On the other hand, by applying the union bound to the bound (\ref{eq:single_wrong_sa}),
	\begin{equation*}
		\mathbb{E}\left[ N_{\ell }\right] \leq 2|\mathcal{A}_{\ell }^{\prime
		}|\max_{i\in \mathcal{A}_{\ell }^{\prime }}\exp \left( -\frac{T2^{\ell
				-1}\Delta _{i}^{2}}{32B^{2}
			K}\right) .
	\end{equation*}%
		Let $i_{\ell }^{\prime }\triangleq \min \mathcal{A}_{\ell }^{\prime }$ $\geq\lfloor\left\vert
	\mathcal{A}_{\ell }\right\vert /4\rfloor +1 \geq K2^{-1-\ell }.$
	Then, we have%
	\begin{equation}
		\mathbb{E}\left[ N_{\ell }\right] \leq 2|\mathcal{A}_{\ell }^{\prime
		}|\max_{i\in \mathcal{A}_{\ell }^{\prime }}\exp \left( -\frac{T2^{\ell
				+1}\Delta _{i}^{2}}{128B^{2}\lfloor\log _{2}\left( K\right)\rfloor K}\right) \leq 2|%
		\mathcal{A}_{\ell }^{\prime }|\exp \left( -\frac{T\Delta _{i_{\ell }^{\prime
			}}^{2}}{128B^{2}\lfloor\log _{2}\left( K\right)\rfloor i_{\ell }^{\prime }}\right) .
		\label{ineq:N_LSA}
	\end{equation}%
	By combining (\ref{eqn:N_lSA}) and (\ref{ineq:N_LSA}), we have
	$$
		\mathbb{P}(1 \notin \mathcal{A}_{\ell +1}|1\in \mathcal{A}_{\ell })\leq
		\mathbb{P}\left (N_{\ell }\geq \frac{1}{3}|\mathcal{A}_{\ell }^{\prime }|\right ) \leq
		8\exp
		\left( -\frac{T\Delta _{i_{\ell }^{\prime }}^{2}}{128B^{2}\lfloor \log _{2}\left(
			K\right)\rfloor i_{\ell }^{\prime }}\right) .
	$$
	Finally, by adding the above inequality for all $\ell $ from $1$ to L, we
	have
$$
		\mathbb{P}(1 \notin \mathcal{A}_{L+1})\leq \sum_{\ell =1}^{L}\mathbb{P}%
		(1\notin \mathcal{A}_{\ell +1}|1\in \mathcal{A}_{\ell })
		\leq8\lfloor\log _{2}(K)\rfloor\exp \left( -\frac{N}{128B^{2}\lfloor \log _{2}\left( K\right)\rfloor
			H_{2}(G)}\right) ,
$$
	where $H_{2}(G)=\max_{i>1}\frac{i}{\Delta _{i}^{2}}.$

	$\square$
\end{proof}
\subsection{Proof of results in Section \ref{sec:app}}
\label{proof:lma:ocba}
\begin{proof}
{Proof of Lemma \ref{lma:ocba_d}} : We consider $\mathcal{X} \subset \mathbb{R}$. Without loss of generality, we assume  $x_1 <x_2<\ldots<x_d$. By the pigeonhole principle, there exists a pair of adjacent discretization points whose distance is no less than $D_\mathcal{X}/(d+1)$  Then,  we assume $x_{j+1}-x_j \geq D_\mathcal{X}/(d+1)$ Then, we construct $K$ systems, where the first and the second systems are identical on the region $((-\infty,x_j]\cup [x_{j+1},+\infty))\cap \mathcal{X}$. Therefore, for $x \in ((-\infty,x_j]\cup [x_{j+1},+\infty))\cap \mathcal{X}$, we assume 
\[
F_i(x,\xi) = -M|x - (x_j + x_{j+1})/2| +\xi, \text{ for } i =1,2,
\]
where $\xi$ is a common random variable with mean zero and variance $\sigma_F^2$. And for $x\in (x_{j-1},x_j)$, we assume 
\[
F_1(x,\xi) = -M|x - (x_j + x_{j+1})/2| +\xi \text{ and } F_2(x,\xi) = -M| (x_j - x_{j+1})/2|.
\]
In this construction, $F_1$ $F_2$ are concave. $v_3,v_4,\ldots,v_K$ could be arbitrary small such that for any $\varepsilon >0$, the complexity $$H <  (M| (x_j - x_{j+1})/2|)^{-2} + \varepsilon \leq \frac{1}{(MD_\mathcal{X}/(d+1) /2)^2}+ \varepsilon.$$
If $d + 1 <\frac{MD_{\mathcal{X}}\sqrt{a}}{2}$, we have $H \leq a$ by taking $\varepsilon$ sufficiently small, which means the constructed instance is in $\mathcal{O}_{M,\mathcal{X},K,a}$. Since the first and the second systems are identical at the discretization set, it is impossible for Algorithm \ref{alg:OCBA} to correctly select the best system.
\end{proof}

\section{Numerical Results}
In this section, we provide additional plots of correct selection regarding different $K$'s, as well as the comparison of optimality gaps.
\newpage
 \subsection{Optimal staffing and pricing in Queueing simulation optimization}
\label{ec:numerical:queue}
This subsection provides further numerical results for $K=16, 24, 32, 40, 48, 128$.
\begin{figure}[htbh]
	\centering
	\subfigure[$K=16$]{
		\label{K=16:q2} \includegraphics[width=1.81in]{PTS_K=16_Queueing.eps}}
	\subfigure[$K=24$]{
		\label{K=24:q2} \includegraphics[width=1.81in]{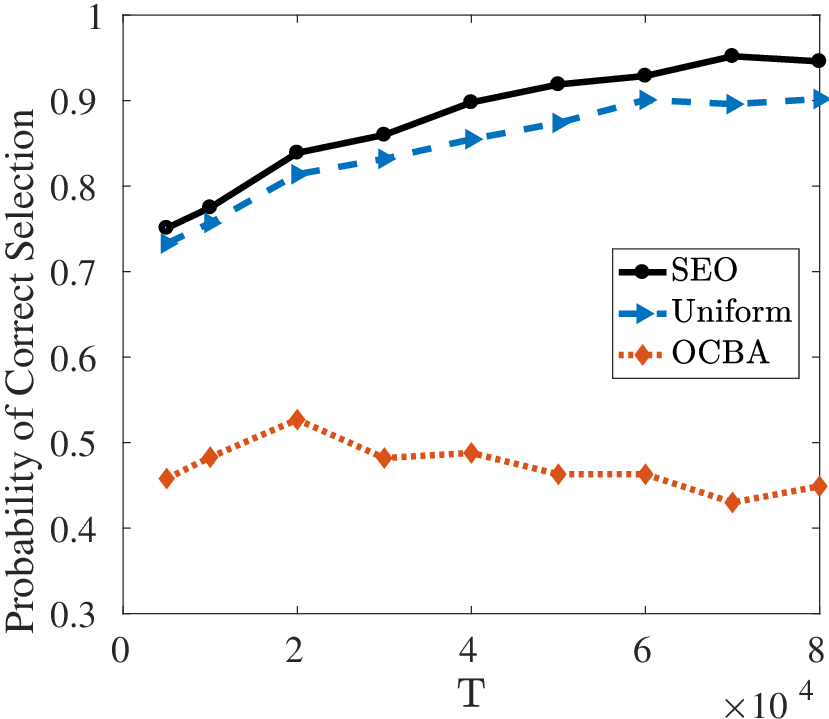}}
	\subfigure[$K=32$]{
		\label{K=32:q2} \includegraphics[width=1.81in]{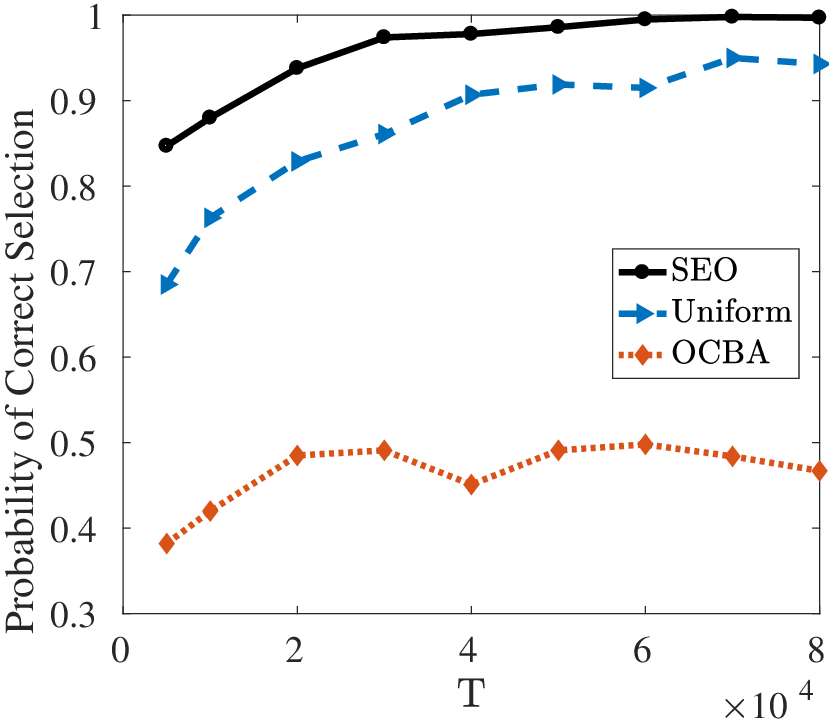}}
	\subfigure[$K=40$]{
		\label{K=40:q2} \includegraphics[width=1.81in]{PTS_K=40_Queueing.eps}}
\subfigure[$K=48$]{
		\label{K=48:q2} \includegraphics[width=1.81in]{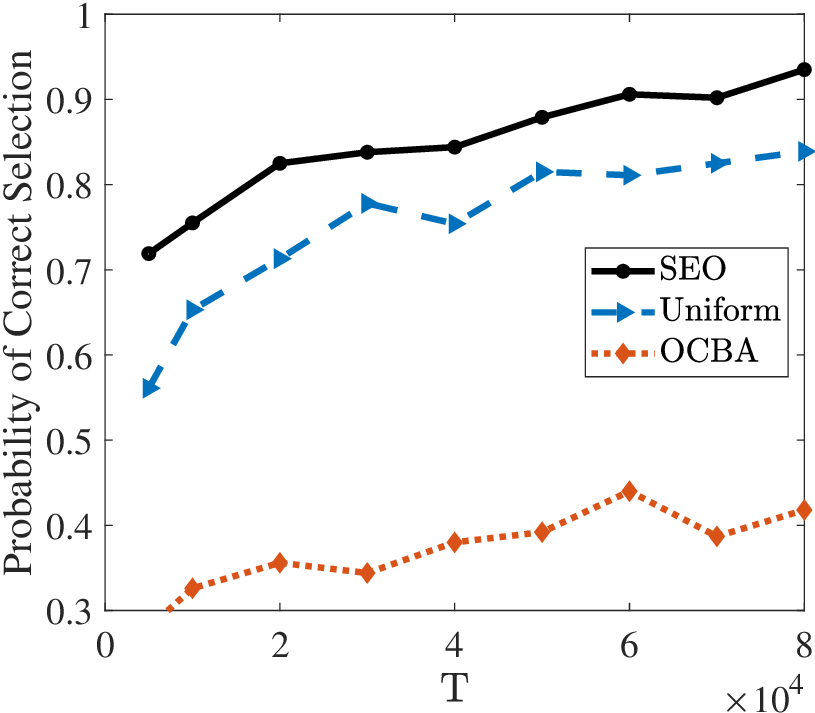}}
	\subfigure[$K=128$]{
		\label{K=128:q2} \includegraphics[width=1.81in]{PTS_K=128_Queueing.eps}}
	\caption{The comparison between SEO, uniform sampling and OCBA in the optimal staffing and pricing problem.}
	\label{fig:queueing2}
\end{figure}

\newpage
 \subsection{Optimal Dosage in the selection of the best drug}
\label{ec:numerical:dose}

We provide the comparison of probability of correct selection and the mean optimality gaps in this subsection. The experiment results show the same pattern we discussed in Section 5. 
\begin{figure}[htbh]
	\centering
	\subfigure[$K=16$]{
		\label{K=16:d2} \includegraphics[width=1.71in]{PTS_K=16_dose.eps}}
	\subfigure[$K=24$]{
		\label{K=24:d2} \includegraphics[width=1.71in]{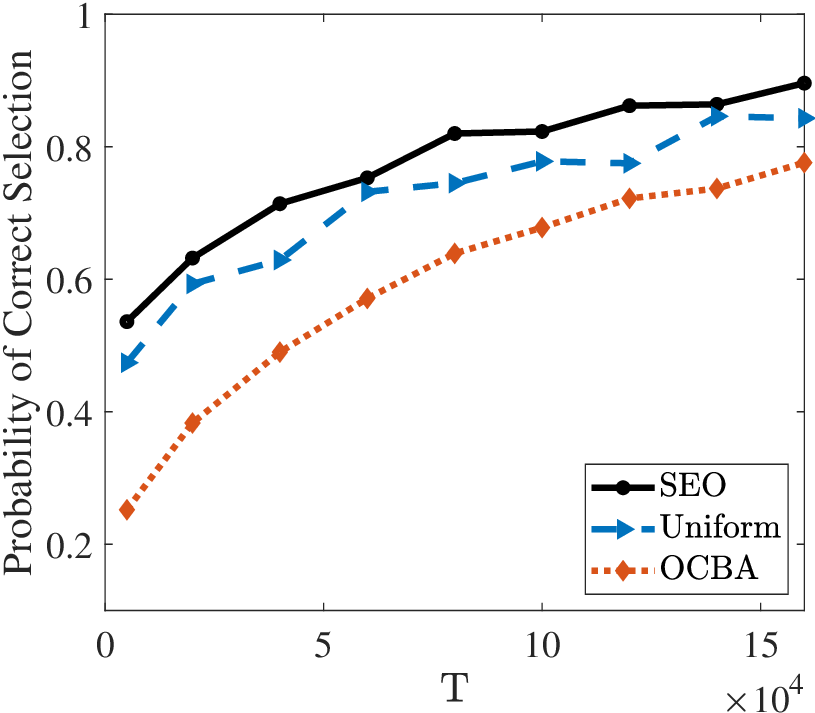}}
	\subfigure[$K=32$]{
		\label{K=32:d2} \includegraphics[width=1.71in]{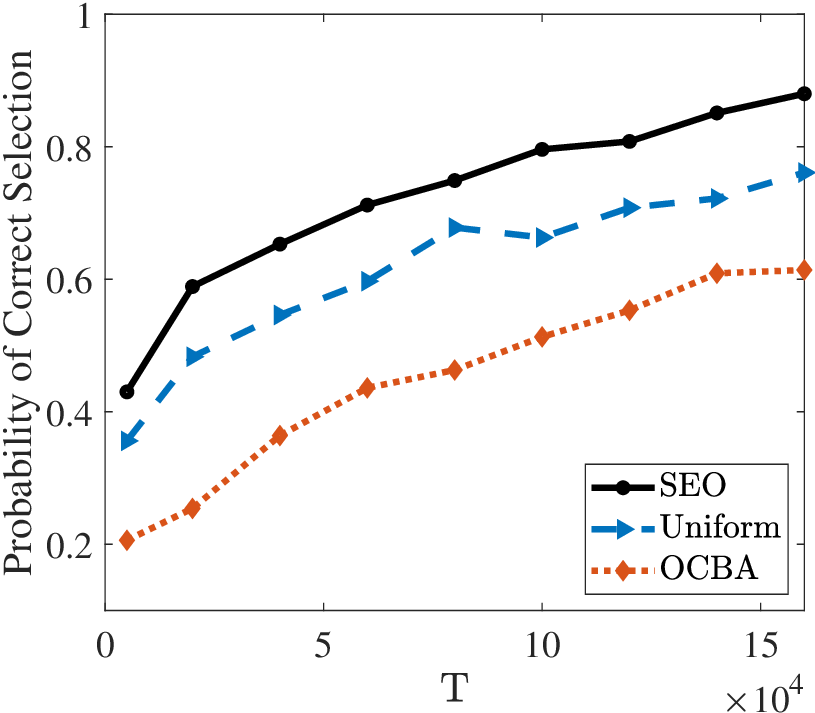}}
	\subfigure[$K=40$]{
		\label{K=40:d2} \includegraphics[width=1.71in]{PTS_K=40_dose.eps}}
\subfigure[$K=48$]{
		\label{K=48:d2} \includegraphics[width=1.71in]{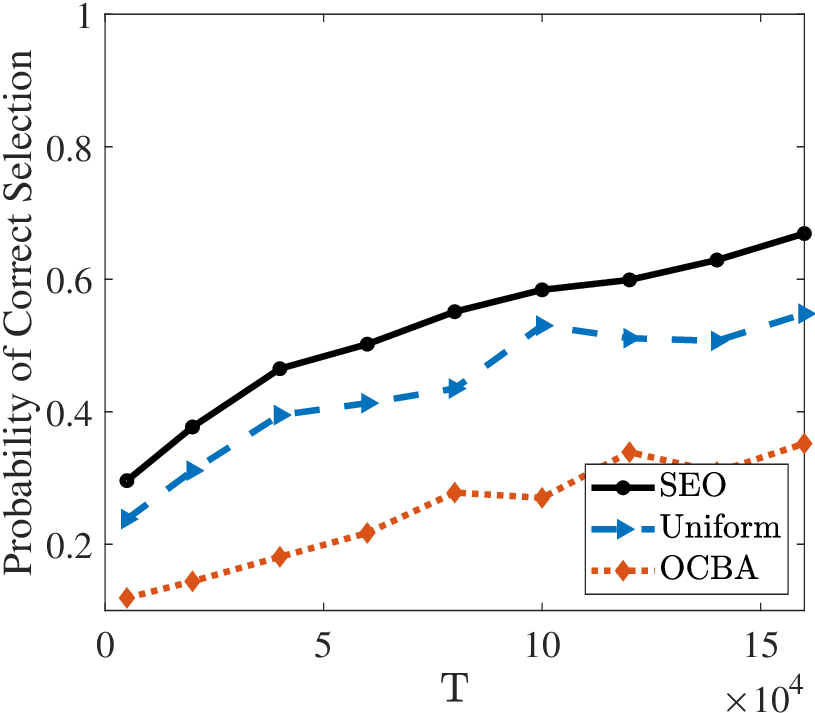}}
	\subfigure[$K=128$]{
		\label{K=128:d2} \includegraphics[width=1.71in]{PTS_K=128_dose.eps}}
	\caption{The comparison between SEO, uniform sampling and OCBA in the optimal dosage problem.}
	\label{fig:dose2}
\end{figure}

\begin{figure}[htbh]
	\centering
	\subfigure[$K=16$]{
		\label{K=16:d2} \includegraphics[width=1.71in]{ pic/g_PTS_K=16_dosage.eps}}
	\subfigure[$K=24$]{
		\label{K=24:d2} \includegraphics[width=1.71in]{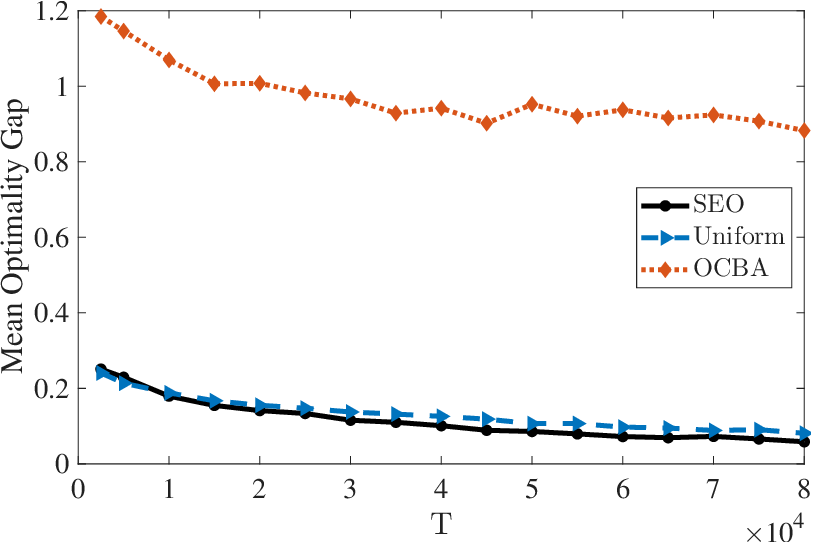}}
	\subfigure[$K=32$]{
		\label{K=32:d2} \includegraphics[width=1.71in]{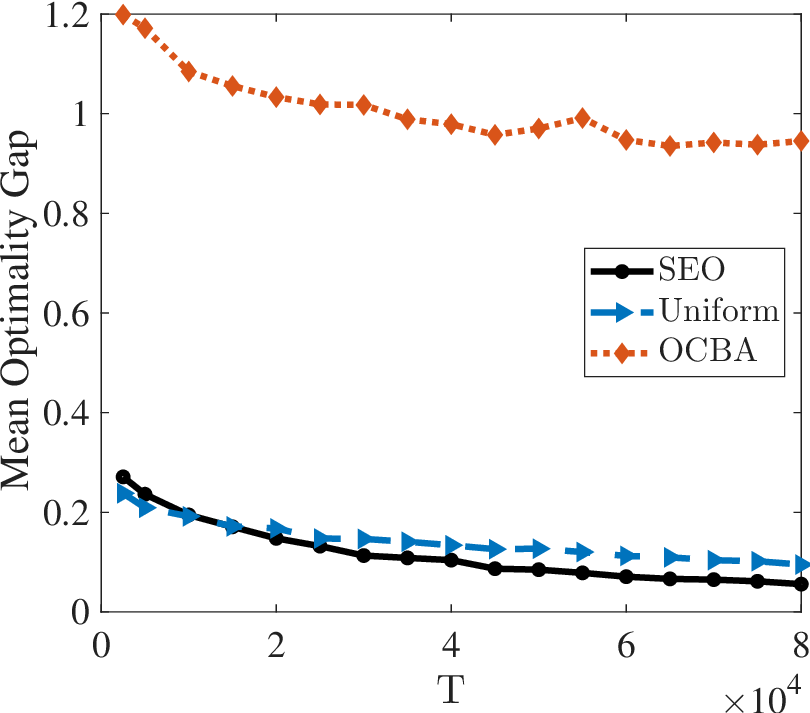}}
	\subfigure[$K=40$]{
		\label{K=40:d2} \includegraphics[width=1.71in]{g_PTS_K=40_dosage.eps}}
\subfigure[$K=48$]{
		\label{K=48:d2} \includegraphics[width=1.71in]{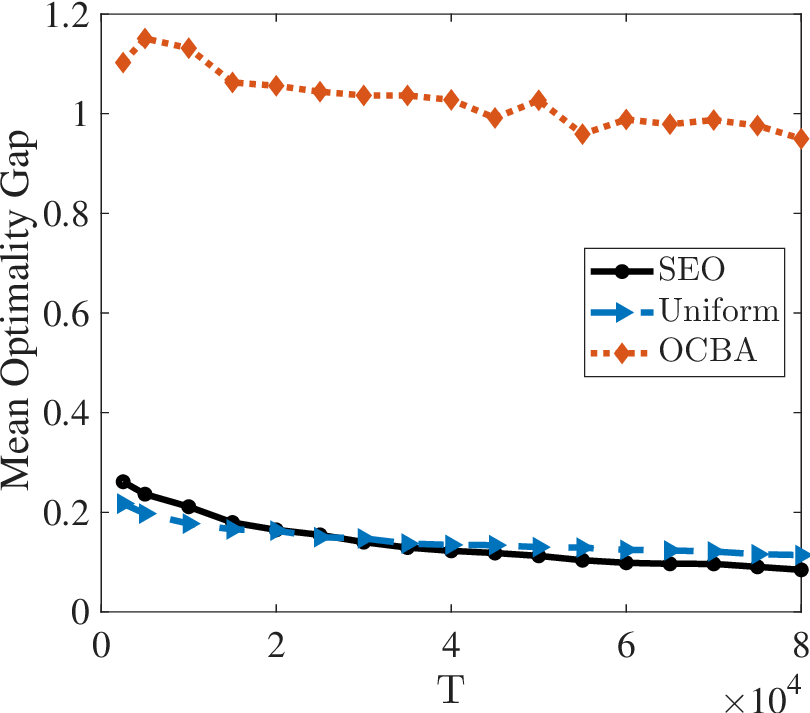}}
	\subfigure[$K=128$]{
		\label{K=128:d2} \includegraphics[width=1.71in]{g_PTS_K=128_dosage.eps}}
	\caption{The comparison between optimality gaps of SEO, uniform sampling, and OCBA in the optimal dosage problem.}
	\label{fig:dose3}
\end{figure}
\newpage
 \subsection{Newsvendor Problem in the selection of the best product}
\label{ec:numerical:newsvendor}
\begin{figure}[htbh]
	\centering
	\subfigure[$K=16$]{
		\label{K=16:n2} \includegraphics[width=1.71in]{PTS_K=16_newsvendor.eps}}
	\subfigure[$K=24$]{
		\label{K=24:n2} \includegraphics[width=1.71in]{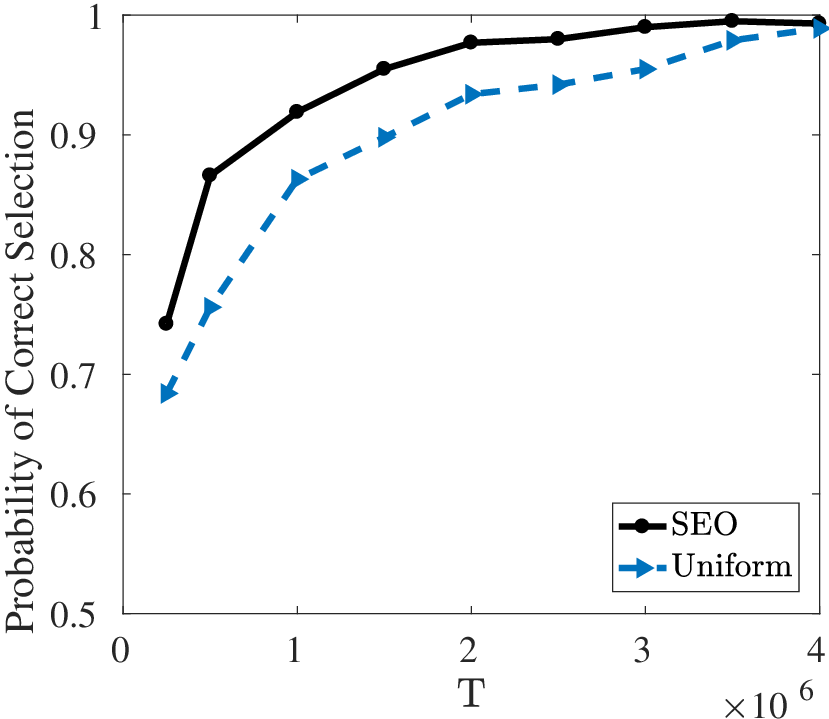}}
	\subfigure[$K=32$]{
		\label{K=32:n2} \includegraphics[width=1.71in]{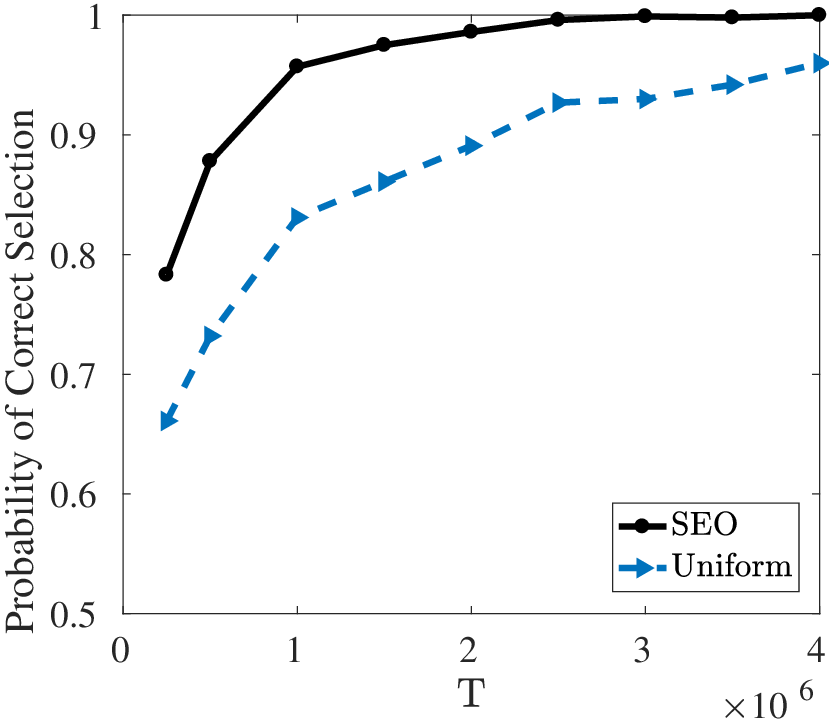}}
	\subfigure[$K=40$]{
		\label{K=40:n2} \includegraphics[width=1.71in]{PTS_K=40_newsvendor.eps}}
	\subfigure[$K=64$]{
		\label{K=64:n2} \includegraphics[width=1.71in]{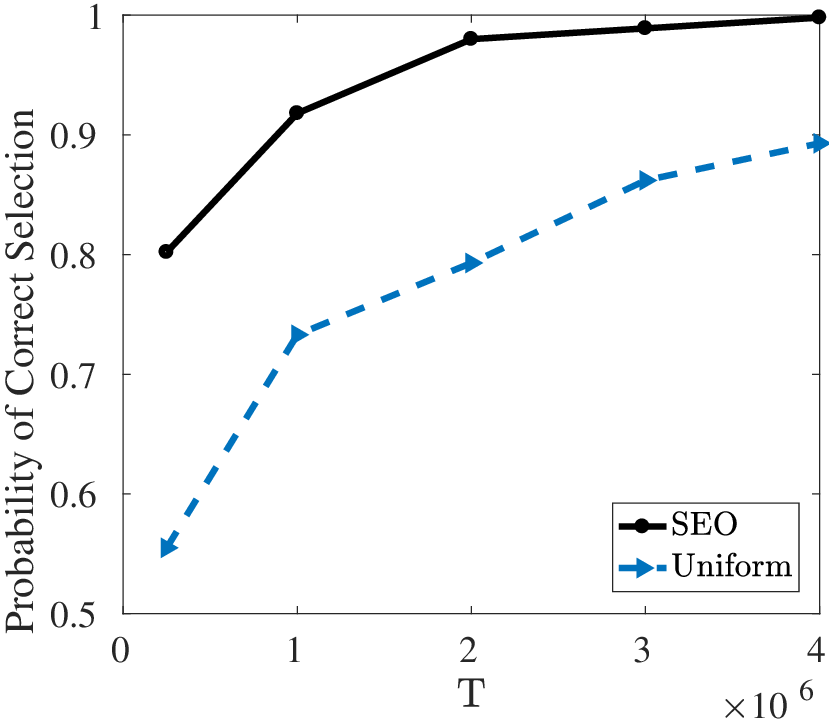}}
	\subfigure[$K=128$]{
		\label{K=128:n2} \includegraphics[width=1.71in]{PTS_K=128_newsvendor.eps}}
	\caption{The comparison between SEO and uniform sampling in the newsvendor problem.}
	\label{fig:news2}
\end{figure}

\begin{figure}[htbh]
	\centering
	\subfigure[$K=16$]{
		\label{K=16:n2} \includegraphics[width=1.81in]{g_PTS_K=16_newsvendor.eps}}
	\subfigure[$K=24$]{
		\label{K=24:n2} \includegraphics[width=1.81in]{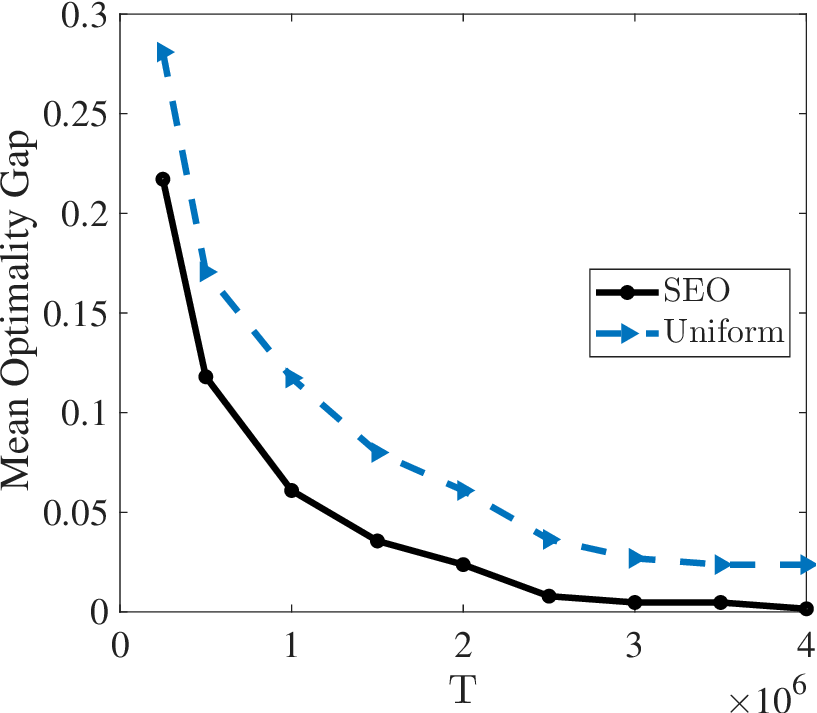}}
	\subfigure[$K=32$]{
		\label{K=32:n2} \includegraphics[width=1.81in]{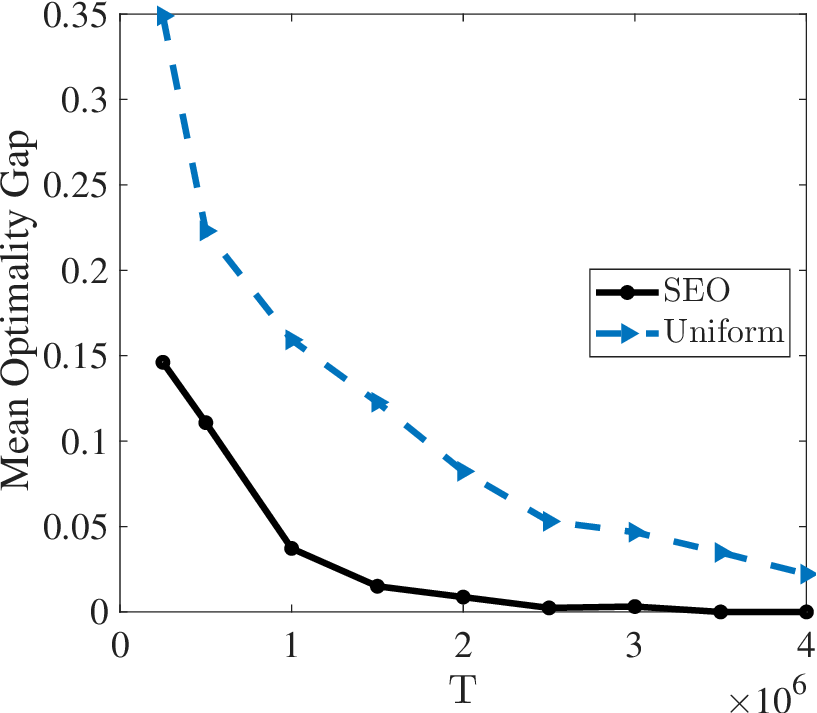}}
	\subfigure[$K=40$]{
		\label{K=40:n2} \includegraphics[width=1.91in]{g_PTS_K=40_newsvendor.eps}}
	\subfigure[$K=48$]{
		\label{K=48:n2} \includegraphics[width=1.91in]{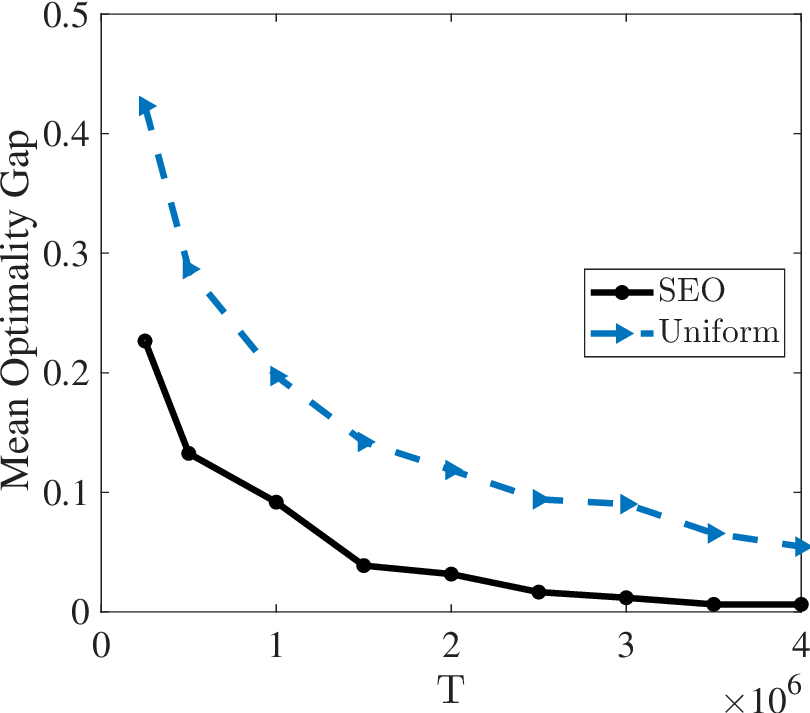}}
	\subfigure[$K=128$]{
		\label{K=128:n2} \includegraphics[width=1.81in]{g_PTS_K=128_newsvendor.eps}}
	\caption{The comparison between optimality gaps of SEO and uniform sampling in the newsvendor problem.}
	\label{fig:news3}
\end{figure}

\end{document}